\documentclass[a4paper,12pt,oneside]{article}

\usepackage{lineno} 
\usepackage{amsmath, amsthm, amscd, amsfonts, amssymb, graphicx}
\usepackage{array,multirow}
\usepackage{float}
\usepackage[a4paper]{geometry}

\usepackage[affil-it]{authblk}
\usepackage{mathrsfs}
\usepackage[dvipsnames,table,xcdraw]{xcolor}
\usepackage{xspace}
\usepackage{subfig}
\usepackage{booktabs,dcolumn}
\usepackage[colorlinks=true]{hyperref}
\usepackage{soul}
\usepackage{multicol}
\usepackage{setspace}

\usepackage{pifont} 
\newcommand{\cmark}{\ding{51}}%
\newcommand{\xmark}{\ding{55}}%

\newtheorem{thm}{Theorem}[section]

\newtheorem{defn}{Definition}[section]

\newtheorem{lemma}{Lemma}[section]
\newtheorem{rem}{Remark}[section]

\title{
Stability of Fractional-Order Discrete-Time Systems with Application to Rulkov Neural Networks and  Asymmetric Memristor Synapses
}

\author[1]{Leila Eftekhari}
\author[2]{Moein Khalighi}
\author[1]{Saeid Abbasbandy}
\affil[1]{Department of Applied Mathematics, Faculty of Science, Imam Khomeini International University,  Qazvin, Iran}
\affil[2]{Department of Computing, University of Turku; 20500 Turku, Finland}

\begin{document}
\maketitle

\begin{abstract}
Memristors have emerged as ideal components for modeling synaptic connections in neural networks due to their ability to emulate synaptic plasticity and memory effects. Discrete models of memristor-coupled neurons are crucial for simplifying computations and efficiently analyzing large-scale neural networks. Furthermore, incorporating fractional-order calculus into discrete models enhances their capacity to capture the memory and hereditary properties inherent in biological neurons, thus reducing numerical discretization errors compared to integer-order models. Despite this potential, discrete fractional-order neural models coupled through memristors have received limited attention.
To address this gap, we introduce two novel discrete fractional-order neural systems. The first system consists of two Rulkov neurons coupled via dual memristors to emulate synaptic functions. The second system expands this configuration into a ring-shaped network of neurons consisting of multiple similar subnetworks. We present a novel theorem that defines stability regions for discrete fractional-order systems, applicable to both proposed models. Integrating discrete fractional-order calculus into memristor-coupled neural models provides a foundation for more accurate and efficient simulations of neural dynamics. This work advances the understanding of neural network stability and paves the way for future research into efficient neural computations.

\end{abstract}

\section{Introduction}

Artificial neural networks and mathematical models simulate the brain's complex activities, offering deeper insights into neural mechanisms~\cite{li2022synchronization}. Neurons do not function in isolation; their interactions are essential for neural processing. Therefore, modeling these neuronal connections using physical devices is critical in network modeling. Due to their nonlinear behavior, nanoscale dimensions, and memory characteristics, memristors are regarded as artificial synapses widely utilized in chaotic circuits and neural networks~\cite{yan2024dynamics}. Memristors demonstrate effective synaptic plasticity, making them highly suitable as emulators of real synapses~\cite{ma2023hidden}. Furthermore, coupling discrete dual memristors instead of a single memristor in neural maps enhances the exploration of complex neuronal behaviors~\cite{ma2023hidden}.

Neuronal models are typically categorized into continuous-time and discrete-time models~\cite{vivekanandhan2023dynamic}. Many well-known models that describe neuronal dynamics are based on ordinary differential equations, such as the Morris-Lecar and FitzHugh-Nagumo models~\cite{yan2024dynamics}. Alternatively, several neuron models are formulated as discrete-time dynamical systems, including the Izhikevich, Rulkov, Chialvo and Courbage-Nekorkin-Vdovin  models~\cite{sun2016complete}. Discrete-time models offer computational advantages and are well-suited for digital hardware implementations, aligning effectively with the integration of memristors as artificial synapses.

While continuous neural models have been extensively studied, discrete models can more effectively simulate random and paroxysmal intracellular ion dynamics~\cite{yan2024dynamics}. They replicate neuronal firing patterns using lower-dimensional systems, which are computationally efficient for large-scale networks and practical in applied sciences. For instance, Peng et al.\cite{peng2020discrete} introduced a discrete memristor and demonstrated that integrating it into the Hénon map enhances performance, yielding a wider range of high complexity suitable for engineering applications such as secure communication. Chen et al.\cite{chen2019flux} developed frameworks for memristor emulators and a fifth-order Chua's circuit with memristive properties, investigating complex and highly sensitive extreme multistability behaviors dependent on initial conditions. Ma et al.~\cite{ma2023multistability} examined the firing patterns and phase synchronization of two Rulkov neurons coupled through a locally active discrete memristor, revealing anti-phase synchronization in similar neurons exhibiting diverse firing dynamics.

Fractional differential equations and differences have recently attracted significant attention. Due to their nonlocal property, fractional derivatives reflect the history of processes, offering enhanced accuracy in modeling real-world phenomena. Fractional calculus has opened new avenues in fields such as electric  systems~\cite{eftekhari2023stability}, biology~\cite{amirian2022extending,khalighi2021three}, ecology~\cite{khalighi2022quantifying}, and neural networks~\cite{kaslik2017dynamics,kaslik2011dynamics,almatroud2021extreme}. Furthermore, fractional-order discrete-time systems prevent numerical discretization errors associated with continuous fractional systems and exhibit more complex dynamics~\cite{almatroud2021extreme,vivekanandhan2023dynamic}. In discrete fractional calculus~\cite{jonnalagadda2015analysis}, two primary approaches are the fractional delta difference ($\Delta$) and the fractional nabla difference ($\nabla$). One drawback of the delta approach is the domain shift that occurs when converting a function to its delta fractional difference, an issue less pronounced with the nabla approach~\cite{goodrich2015discrete}.

Applying fractional-order calculus to neural modeling yields superior nonlinear features and richer dynamic behaviors compared to integer-order models. Lu et al.\cite{lu2022dynamics} introduced a fractional-order discrete memristor model and developed a memristor-driven fractional Rulkov neuron map. Their fractional-order model demonstrates complex dynamics, such as asymptotically periodic oscillations, hyperchaos, multistability, and transient chaos, highlighting the significance of fractional calculus in capturing the memory effects inherent in neural systems. Furthermore, modeling neuronal activity using fractional-order dynamics captures history-dependent behaviors, resulting in power-law, non-Markovian dynamics in membrane voltage\cite{weinberg2017history}.

Discrete fractional calculus has effectively modeled real-world problems across diverse fields, including tumor-immune interactions~\cite{alzabut2023existence}, COVID-19 pandemic models~\cite{abbes2023fractional}, and image processing~\cite{ostalczyk2015discrete}. In neural dynamics, discrete fractional models exhibit more complex behaviors than their continuous counterparts. For instance, the discrete fractional Morris–Lecar model transitions from regular spiking to sequences of rapid spikes and inconsistent abundance~\cite{chu2023advanced}.

Despite these advancements, existing studies on discrete memristive neural networks primarily consider individual memristors coupled to discrete neurons, which may not fully capture the complexity of biological neuronal dynamics~\cite{ma2023hidden}. Coupling neurons via dual memristors can produce complex behaviors such as extreme transient dynamics, aligning more closely with real neuronal operations~\cite{du2021coupling}. However, the incorporation of fractional-order discrete calculus, especially using the nabla operator, into dual memristor-coupled neural networks remains underexplored.

To address this gap, we propose a novel exploration of fractional-order discrete memristor-coupled neural networks. Specifically, we construct a Fractional Discrete Memristor-Coupled Rulkov Neuron (FDMCRN) map by coupling two Rulkov neurons through two discrete flux-controlled memristors functioning as synapses. We employ the nabla fractional-order Caputo difference to model the fractional dynamics. We then expand this model into a ring network configuration of FDMCRN maps, establishing theoretical foundations for stability analysis of the ring structure, which is essential for understanding long-term neural network behavior. Finally, we numerically investigate stability criteria for both systems and assess how initial conditions influence the dynamics of the FDMCRN model.

\section{Materials and Methods}

\begin{rem}
In this article, we consistently work within the framework of a discrete time scale,  $ \mathbb{N}_b = \{ b,b + 1,b + 2,...\} $ where $b\in \mathbb{R}$ is fixed. For any function $f:\mathbb{N}_b \to \mathbb{R}$, the backward difference also known as the nabla operator is defined as  $\nabla f(t) = f(t) - f(t - 1)$ for $t \in \mathbb{N}_{b+1}$.
\end{rem}

\begin{defn}\label{def:rising}
\cite{jonnalagadda2016solutions}
 The definition of the $\tau $ rising function is given as follows
\begin{align}
t\overline {^\tau }  = \frac{{\Gamma (t + \tau )}}{{\Gamma (t)}}, \quad {0^{\overline \tau  }} = 0,
\end{align}
 for all $ \tau, t \in \mathbb{R}$.
\end{defn}

\begin{defn} \cite{mohan2014stability}
\label{def:cap}
(Nabla fractional sum)
Suppose $f:\mathbb{N}_0 \to \mathbb{R} $ and $\alpha \in (0,1)$ be provided, the nabla fractional sum of order  ${\alpha }$ for the function $f$ is expressed as follows
\begin{align}
\label{f1}
{\nabla }^{ - \alpha }f(x) = \frac{1}{{\Gamma (\alpha )}}\sum\limits_{s = 1}^x {{{(x - \rho (s))}^{\overline {\alpha  - 1} }},}
 \end{align}
where $\rho (s) = s - 1$,   $x\in \mathbb{N}_0$.
 \end{defn}

\begin{defn} \cite{jonnalagadda2016solutions}
\label{def:cap2}
Let $f:\mathbb{N}_0 \to \mathbb{ R} $ and $\alpha \in (0,1)$, the  Caputo-type fractional nabla difference of order ${\alpha }$ for $f$ is
\begin{align}
\label{f1b}
{}_0^C\nabla _{}^\alpha f(x): = ({}_0 \nabla ^{ - (1 - \alpha )}\nabla f)(x) = \frac{1}{{\Gamma (1 - \alpha )}}\sum\limits_{s = a + 1}^x {{{(x - s + 1)}^{\overline { - \alpha } }}(\nabla f)(s),}
 \end{align}
  for $\alpha  = 0$, we get ~$\nabla _0^0f(x) = f(x)$  and  $x\in \mathbb{N}_1$.
 \end{defn}

\subsection{Mathematical  Representations}

\textbf{\textsc{FDMCRN} Framework}.
Using the definition (\ref{f1}) and the mathematical expression of the discrete-time memristor \cite{liu2022dynamics, li2022synchronization,yan2024dynamics}, the difference equation of the Rulkov neuron map can be derived the following form

\begin{align}\label{eq:fmhnn_2_N}
\begin{cases}
{}_0^C{\nabla ^\alpha }{x_{1,n + 1}} = \frac{a}{{1 + {x^2}{{_{1,n}}^{}}}} + {y_{1,n}} + k\varphi_{1,n} ({x_{1,n}} - {x_{2,n}}) - {x_{1,n}},\\
{}_0^C{\nabla ^\alpha }{y_{1,n + 1}} =  - \eta  ( {x_{1,n}} - \sigma ) + k\varphi_{2,n} ({y_{1,n}} - {y_{2,n}}),\\
{}_0^C{\nabla ^\alpha }{x_{2,n + 1}} = \frac{a}{{1 + {x^2}{{_{2,n}}^{}}}} + {y_{2,n}} - k\varphi_{1,n} ({x_{1,n}} - {x_{2,n}}) - {x_{2,n}},\\
{}_0^C{\nabla ^\alpha }{y_{2,n + 1}} =  - \eta  ( {x_{2,n}} - \sigma ) - k\varphi_{2,n} ({y_{1,n}} - {y_{2,n}}),\\
{}_0^C{\nabla ^\alpha }{\varphi _{1,n + 1}} = {x_{1,n}} - {x_{2,n}},\\
{}_0^C{\nabla ^\alpha }{\varphi _{2,n + 1}} = {y_{1,n}} - {y_{2,n}},\\
\end{cases}
\end{align}
where $n$ denotes the iterations and $k$ is the coupling strength between the pair of neurons. The variables $x_{1,n}, x_{2,n}$  represent the neurons' membrane potentials influenced by a nonlinear map parameter $a$. The recovery variables $y_{1,n}, y_{2,n}$ are influenced by parameters $ 0 < \eta  <<  1$ and $\sigma $, which is an external DC to neurons. Additionally, $\varphi _{1,n}, \varphi _{2,n }$ are flux variables of the neurons at the $n$-th iteration. The schematic connection pattern of the \textsc{FDMCRN}, consisting of two neurons with two coupling memristors, is depicted in Fig.~\ref{fig1}.
\begin{figure}[!tbp]
\begin{center}
\includegraphics[width=12cm, height=6cm]{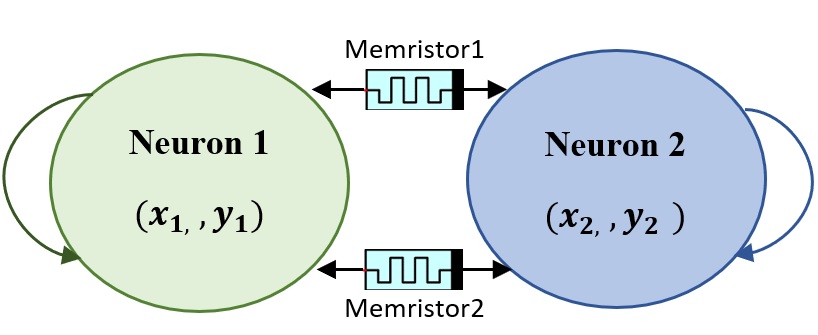}
\end{center}
\caption{Connection pattern of the \textsc{FDMCRN} model.}\label{fig1}
\end{figure}
\\
\\
\textbf{\textsc{ FDMCRN } Framework with Ring  Configuration}.
We enhance the  \textsc{FDMCRN} model (\ref{eq:fmhnn_2_N}) by increasing the number of neurons from two to $N$ and implementing a ring-network structure, as illustrated in Fig.~\ref{fig2}. In this structure, $N$ units are connected through coupling, with the membrane potentials $x_{1i}$ evolving according to the following equations
\begin{align}\label{ndim}
\begin{cases}
{}_0^C{\nabla ^\alpha }{x_{1i,n + 1}} = \frac{a}{{1 + {x^2}{{_{1i,n}}^{}}}} + {y_{1i,n}} + k\varphi_{1i,n} ({x_{1i,n}} - {x_{2i,n}}) - {x_{1i,n}}+ \frac{D}{{2P}}\sum\limits_{j = i - P}^{i + P} {({x_{1j}} - {x_{1i}})},\\
{}_0^C{\nabla ^\alpha }{y_{1i,n + 1}} =  - \eta ( {x_{1i,n}} - \sigma ) + k\varphi_{2i,n} ({y_{1i,n}} - {y_{2i,n}}),\\
{}_0^C{\nabla ^\alpha }{x_{2i,n + 1}} = \frac{a}{{1 + {x^2}{{_{2i,n}}^{}}}} + {y_{2i,n}} - k\varphi _{1i,n}({x_{1i,n}} - {x_{2i,n}}) - {x_{2i,n}},\\
{}_0^C{\nabla ^\alpha }{y_{2i,n + 1}} =  - \eta ( {x_{2i,n}} - \sigma ) - k\varphi_{2i,n} ({y_{1i,n}} - {y_{2i,n}}),\\
{}_0^C{\nabla ^\alpha }{\varphi _{1i,n + 1}} = {x_{1i,n}} - {x_{2i,n}},\\
{}_0^C{\nabla ^\alpha }{\varphi _{2i,n + 1}} = {y_{1i,n}} - {y_{2i,n}},\\
\end{cases}
\end{align}
where $i = 1$ to $N$ and the state variables of the $i^{\text{th}}$ unit are represented as $(x_{1i}, y_{1i}, x_{2i}, y_{2i}, \varphi_{1i}, \varphi_{2i})$.
\\The memristor synapses couple symmetrically to the $2P$ nearest neighbors, every unit engages with $P$ neighboring units to its left and right. Here, $P$ denotes the set of indices corresponding to the neighbors of the $i^{\text{th}}$ unit and $D$ represents the connection strength between  the $N$ units.
\begin{figure}[!tbp]
	\begin{center}
\includegraphics[width=15cm, height=9cm]{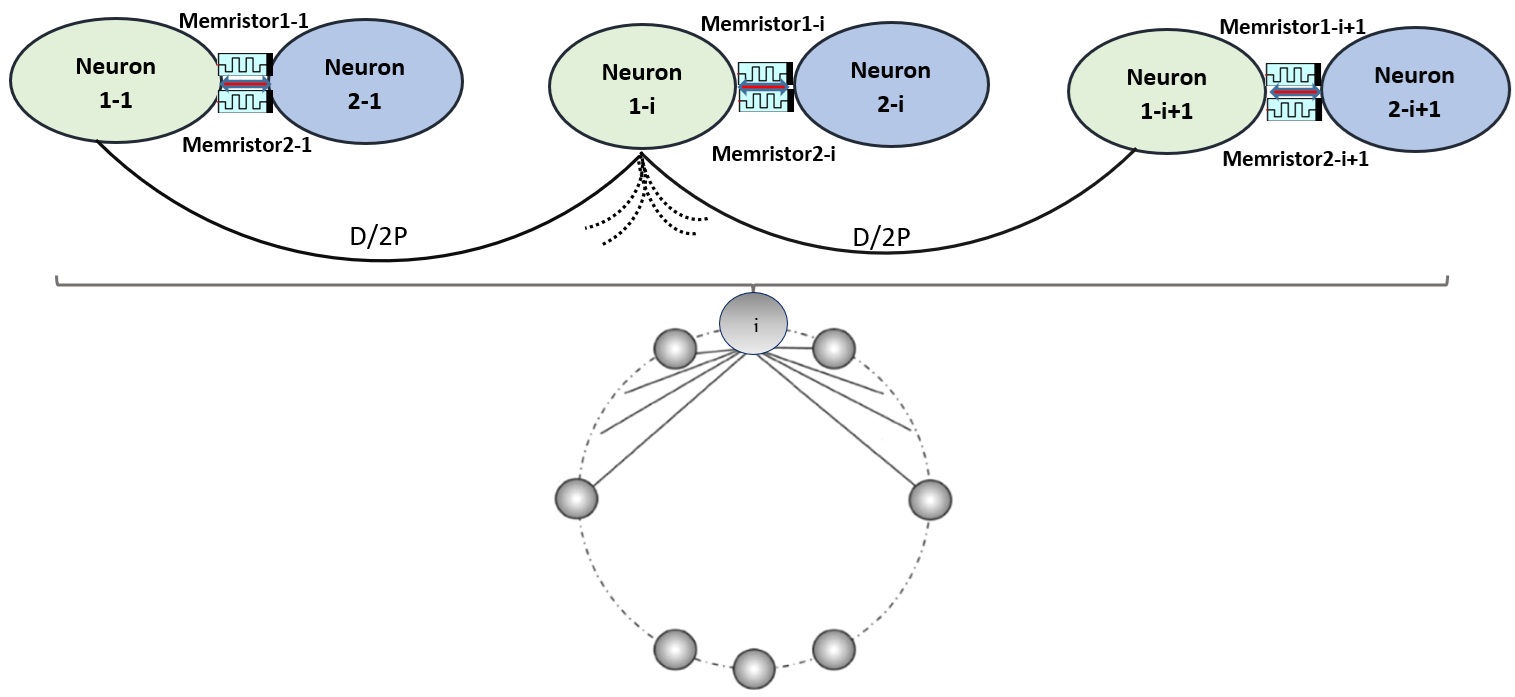}
\end{center}
\caption{The i-th unit of the ring network including the \textsc{FDMCRN} model.}\label{fig2}
\end{figure}

\subsection{Analysis of the Model}
This section explores the existence and uniqueness of solutions concerning stability conditions for models \eqref{eq:fmhnn_2_N} and \eqref{ndim}. Each subsection begins with necessary definitions, lemmas, and theorems.
\subsubsection{Analysis of Possible Outcomes }
\begin{defn}
A \textbf{circulant matrix} is a structured square matrix where each row is defined as
\begin{equation*}
C = circ (b_0, b_1, . . . , b_{N-1})=
\left[
{\begin{array}{*{20}{c}}
c_{0}    & c_{1} & c_{2} & c_{3} & \ddots & c_{N-2} & c_{N-1}\\
c_{N-1} & c_{0} & c_{1} & c_{2} & \ddots & c_{N-3}& c_{N-2}\\
c_{N-2} & c_{0} & c_{1} & c_{2} & \ddots & c_{N-4} & c_{N-3}\\
\ddots & \ddots & \ddots & \ddots &  \ddots & \ddots & \ddots \\
c_{2}     & c_{3} & c_{4} & c_{5} & \ddots & c_{0}  & c_{1} \\
c_{1}    & c_{2}  & c_{3} & c_{4} & \ddots & c_{N-1} & c_{0}
\end{array}} \right]_{N \times N}.
\end{equation*}
A \textbf{block circulant matrix} extends the concept of a circulant matrix by having each block as an individual circulant matrix, shown below
\begin{equation*}
C = bcirc (C_0, C_1, . . . , C_{N-1})=
\left[
{\begin{array}{*{20}{c}}
C_{0}    & C_{1} & C_{2} & C_{3} & \ddots & C_{N-2} & C_{N-1}\\
C_{N-1} & C_{0} & C_{1} & C_{2} & \ddots & C_{N-3}& C_{N-2}\\
C_{N-2} & C_{0} & C_{1} & C_{2} & \ddots & C_{N-4} & C_{N-3}\\
\ddots & \ddots & \ddots & \ddots &  \ddots &\ddots & \ddots \\
C_{2}     & C_{3} & C_{4} & C_{5} & \ddots & C_{0}  & C_{1} \\
C_{1}    & C_{2}  & C_{3} & C_{4} & \ddots & C_{N-1} & C_{0}
\end{array}} \right]_{N \times N},
\end{equation*}
the $C_i$ matrices for  $i=0$ to $N-1$ are block-formatted matrices \cite{eftekhari2023stability}.
\end{defn}

\begin{thm}\label{thr1}
Suppose that \\
$\Omega=\{(x_1, y_1,x_2,y_2,\varphi_1,\varphi_2) \in \mathop{R}^6$, $ \max\{||x_i||,||y_i||, ||\varphi_i||\} \leq \zeta \}$, i=1,2 and $\textsc{S}=\Omega \times \mathbb{ N}_1$ where $n<\infty$. For any initial conditions $\textsc{X}_0=(x_{1,0}, y_{1,0}, x_{2,0}, y_{2,0}, \varphi_{1,0}, \varphi_{2,0}) \in \Omega$, for every n, all  the solutions $\textsc{X}_n \in\textsc{S} $ of system (\ref{eq:fmhnn_2_N}) are  unique.
\end{thm}

\begin{proof}
Assume $\mathbb{F} (\textsc{X}_n)=(\mathbb{F}_1(\textsc{X}_n), \mathbb{F}_2(\textsc{X}_n), \mathbb{F}_3(\textsc{X}_n),\mathbb{F}_4(\textsc{X}_n),\mathbb{F}_5(\textsc{X}_n),\mathbb{F}_6(\textsc{X}_n))^T$ represents a mapping function with the $||.||$ norm, ensuring that
\begin{align}
    \mathbb{F}_1(\textsc{X}_n)&=\frac{a}{{1 + x_{1,n}^2}} + {y_{1,n}} + k\varphi_{1,n} ({x_{1,n}} - {x_{2,n}}) - {x_{1,n}},\nonumber\\
    \mathbb{F}_2(\textsc{X}_n)&= - \eta (1 + {x_{1,n}} - \sigma ) + k\varphi_{2,n} ({y_{1,n}} - {y_{2,n}}),\nonumber\\
    \mathbb{F}_3(\textsc{X}_n)&=\frac{a}{{1 + x_{2,n}^2}} + {y_{2,n}} - k\varphi_{1,n} ({x_{1,n}} - {x_{2,n}}) - {x_{2,n}},\nonumber\\
    \mathbb{F}_4(\textsc{X}_n)&=  - \eta (1 + {x_{2,n}} - \sigma ) - k\varphi_{2,n} ({y_{1,n}} - {y_{2,n}}),\nonumber\\
    \mathbb{F}_5(\textsc{X}_n)&= {x_{1,n}} - {x_{2,n}},\nonumber\\
    \mathbb{F}_6(\textsc{X}_n)&={y_{1,n}} - {y_{2,n}},  \nonumber
\end{align}
where $\textsc{X}_n=(x_{1,n},y_{1,n}, x_{2,n},y_{2,n}, \varphi_{1,n},\varphi_{2,n})^T$. Consequently, the mapping function can be expressed in the form below 
$$\mathbb{F}(\textsc{X}_n)= \Psi (\textsc{X}_n) + A \textsc{\textsc{X}}_n + k \Phi \textsc{X}_n ,$$ such that
\begin{align*}
\Psi (X_n)=
\left[\begin{array}{*{20}{c}}
{\frac{a}{{1 + x_{1,n}^2}}}\\
{\eta \sigma }\\
{\frac{a}{{1 + x_{2,n}^2}}}\\
{\eta \sigma }\\
0\\
0
\end{array}\right],
 A=
\left[\begin{array}{*{20}{c}}
{ - 1}&1&0&0&0&0\\
{ - \eta }&0&0&0&0&0\\
0&0&{ - 1}&1&0&0\\
0&0&{ - \eta }&0&0&0\\
1&0&{ - 1}&0&0&0\\
0&1&0&{ - 1}&0&0
\end{array}\right], ~~
\\
\Phi=
\left[\begin{array}{*{20}{c}}
{{\varphi _{1,n}}}&0&{ - {\varphi _{1,n}}}&0&0&0\\
0&{{\varphi _{2,n}}}&0&{ - {\varphi _{2,n}}}&0&0\\
{ - {\varphi _{1,n}}}&0&{{\varphi _{1,n}}}&0&0&0\\
0&{ - {\varphi _{2,n}}}&0&{{\varphi _{2,n}}}&0&0\\
0&0&0&0&0&0\\
0&0&0&0&0&0
\end{array}\right],
\textsc{X}_n =
\left[ \begin{array}{*{20}{c}}
{{x_{1,n}}}\\
{{y_{1,n}}}\\
{{x_{2,n}}}\\
{{y_{2,n}}}\\
{{\varphi _{1,n}}}\\
{{\varphi _{2,n}}}
\end{array} \right],
\end{align*}
where $\Psi(\textsc{x}) $ is continuous. Therefore, based on  the Lipschitz condition \cite{mert2019existence}, for all  $\textsc{X}, \bar{\textsc{X}}$ in $\textsc{S}$, there exists a constant $\zeta_1 \geq 0$, s.t $  || (\Psi(\textsc{x})-\Psi(\bar{\textsc{x}})|| \leq  \zeta _1 ||\textsc{x}-\bar{\textsc{x}}||$.\\
It will now be shown that the system (\ref{eq:fmhnn_2_N}) meets  Lipschitz condition \cite{mert2019existence}, s.t
$    ||\mathbb{F}(\textsc{X})-\mathbb{F}(\bar{\textsc{X}})||  \leq  L||\textsc{X}-\bar{\textsc{X}}|| \nonumber $ Where ${t_1},{t_2} \in {N_1},{t_1} <  t_2$, $L=\left( \zeta _1+ ||A|| + k~{\zeta} \right)>0$. Thus, the solution to the  model (\ref{eq:fmhnn_2_N}) is unique within $\textsc{S}$.
\end{proof}

\begin{thm}\label{th:uniq_ND}
Suppose that ~ $\Omega=\{(\textsc{x}_1, \textsc{y}_1,\textsc{x}_2,\textsc{y}_2, \phi_1, \phi_2) \in \mathbb{R}^{N \times 6}, ~\text{s.t}~~ \textsc{x}_1 = [x_{1i,n}]_{{}_{N \times 1}},~ \textsc{y}_1 = [y_{1i,n}]_{{}_{N \times 1}},
\textsc{x}_2 = [x_{2i,n}]_{{}_{N \times 1}}, \textsc{y}_2 = [y_{2i,n}]_{{}_{N \times 1}}, ~ \phi_1 = [\varphi_{1i,n}]_{{}_{N \times 1}}, ~ \phi_2 = [\varphi_{2i,n}]_{{}_{N \times 1}}, ~ \Lambda = [\lambda_i]_{{}_{N \times 1}}, ~
\max\{||\textsc{x}_1||, ||\textsc{y}_1||, ||\textsc{x}_2||, ||\textsc{y}_2||, ||\phi_1||, ||\phi_2||\} \leq \Lambda\}_{|_{i=1,2, ..., N}}$ and  $S=\Omega \times \mathbb{ N}_1$. For any chosen initial conditions $\left(\textsc{x}_{1,0}, \textsc{y}_{1,0}, \textsc{x}_{2,0},\textsc{y}_{2,0}, \phi_{1,0},\phi_{2,0}\right) \in \Omega$, all the solutions $\textsc{X}_n \in S$ of model (\ref{ndim}) are  unique for all $\textsc{t} \geq 0$.
\end{thm}
\begin{proof}
Suppose $$\textsc{X}_n=(X_{1,n},Y_{1,n}, X_{2,n},Y_{2,n}, \phi_{1,n},\phi_{2,n})^T,$$  $$\mathbb{F} (\textsc{X}_n)=(\mathbb{F}_1(\textsc{X}_n), \mathbb{F}_2(\textsc{X}_n), \mathbb{F}_3(\textsc{X}_n), \mathbb{F}_4(\textsc{X}_n), \mathbb{F}_5(\textsc{X}_n), \mathbb{F}_6(\textsc{X}_n))^T,$$act as a function that maps multiple variables while utilizing the $||.||$ norm in such a way that
\begin{align}
 \mathbb{F}_1(\textsc{X}_n )&=\frac{a}{{1 + x_{1i,n}^2}} + {y_{1i,n}} + k\varphi_{1i,n} ({x_{1i,n}} - {x_{2i,n}}) - {x_{1i,n}}+\frac{D}{{2P}}\sum\limits_{j = i - P}^{i + P} {({x_{1j,n}} - {x_{1i,n}})},\nonumber\\
    \mathbb{F}_2(\textsc{X}_n )&= - \eta (1 + {x_{1i,n}} - \sigma ) + k\varphi_{2i,n} ({y_{1i,n}} - {y_{2i,n}}),\nonumber\\
    \mathbb{F}_3(\textsc{X}_n )&=\frac{a}{{1 + x_{2i,n}^2}} + {y_{2i,n}} - k\varphi_{1i,n} ({x_{1i,n}} - {x_{2i,n}}) - {x_{2i,n}},\nonumber\\
    \mathbb{F}_4(\textsc{X}_n )&=  - \eta (1 + {x_{2i,n}} - \sigma ) - k\varphi_{2i,n} ({y_{1i,n}} - {y_{2i,n}}),\nonumber\\
    \mathbb{F}_5(\textsc{X}_n )&= {x_{1i,n}} - {x_{2i,n}},\nonumber\\
    \mathbb{F}_6(\textsc{X}_n )&={y_{1i,n}} - {y_{2i,n}}.  \nonumber
\end{align}
It can be written as follows
$$\mathbb{F}(\textsc{X}_n )= \Psi (\textsc{X}_n )+A \textsc{X}_n  + k \widetilde \Phi \textsc{X}_n   ~\widetilde \Phi =bcirc(\Phi ,0,..,0)_{1\times N},$$ $$ \Psi (\textsc{X}_n ) =bcirc(\phi(\textsc{X}_n ) ,\phi(\textsc{X}_n ),..,\phi(\textsc{X}_n ))_{1\times N}.$$
Matrix $A \in {\mathbb{R}^{N \times N}}$ is a block circulant matrix constructed from the blocks $ A_0,A_1,...,A_1 $, such that  $A_0,A_1 \in {\mathbb{R}^{6 \times 6}}$,
\begin{align*}
A_0=\left[ \begin{array}{*{20}{c}}
{ - 1 - \frac{D}{P}}&1&0&0&0&0\\
{ - \eta }&0&0&0&0&0\\
0&0&{ - 1}&1&0&0\\
0&0&{ - \eta }&0&0&0\\
1&0&{ - 1}&0&0&0\\
0&1&0&{ - 1}&0&0
\end{array} \right],
A_1=\left[ \begin{array}{*{20}{c}}
{\frac{D}{{2P}}}&0&0&0&0&0\\
0&0&0&0&0&0\\
0&0&0&0&0&0\\
0&0&0&0&0&0\\
0&0&0&0&0&0\\
0&0&0&0&0&0
\end{array} \right],~~~
\end{align*}
\begin{align*}
\Phi=
\left[\begin{array}{*{20}{c}}
{{\varphi _{1,n}}}&0&{ - {\varphi _{1,n}}}&0&0&0\\
0&{{\varphi _{2,n}}}&0&{ - {\varphi _{2,n}}}&0&0\\
{ - {\varphi _{1,n}}}&0&{{\varphi _{1,n}}}&0&0&0\\
0&{ - {\varphi _{2,n}}}&0&{{\varphi _{2,n}}}&0&0\\
0&0&0&0&0&0\\
0&0&0&0&0&0
\end{array}\right], ~~
\textsc{X}_n =
\left[ \begin{array}{*{20}{c}}
{{X_{1,n}}}\\
{{Y_{1,n}}}\\
{{X_{2,n}}}\\
{{Y_{2,n}}}\\
{{\phi _{1,n}}}\\
{{\phi _{2,n}}}
\end{array} \right],
\phi(\textsc{X}_n )=
\left[\begin{array}{*{20}{c}}
{\frac{a}{{1 + x_{1,n}^2}}}\\
{\eta \sigma }\\
{\frac{a}{{1 + x_{2,n}^2}}}\\
{\eta \sigma }\\
0\\
0
\end{array}\right],
\end{align*}
\\
\\
\\
where $\Psi(\textsc{x}) $ is  continuous and consequently satisfies Lipschitz condition \cite{mert2019existence}.\\
The following demonstrates that the system  (\ref{ndim}) satisfies  Lipschitz condition \cite{mert2019existence},
\begin{align*}
   \left\| {F(X) - F(\overline X )} \right\|  \le \ L \left\| {X - \overline X } \right\|,
\end{align*}
where ${t_1},{t_2} \in {N_1},{t_1} < t_2$
and $$L=(||A|| +(2k+1) ~ ||\Lambda||),$$
 $$ \textsc{X}= \left(\textsc{x}_{1,t_1}, \textsc{y}_{1,t_1}, \textsc{x}_{2,t_1},\textsc{y}_{2,t_1}, \phi_{1,t_1},\phi_{2,t_1}\right),$$
$$ \overline  X= \left(\overline x_{1,t_2}, \overline y_{1,t_2}, \overline x_{2,t_2},\overline y_{2,t_2},\overline \phi_{1,t_2},\overline \phi_{2,t_2}\right).$$\\
\\
 This means that  the network model (\ref{ndim}) admits a unique solution  $\textsc{X}_n \in \textsc{S}$.
\end{proof}
\subsubsection{Stability Analysis}
This section provides the theoretical foundation for analyzing stability.

\begin{defn}\label{def:laplace}
\cite{nechvatal2014asymptotics}
The inverse Laplace transform of a function $f:{{\mathbb{N}}} \to \mathbb{C}$ is
\begin{align}
{\aleph_a}[f(t)] = \sum\limits_{k = 1}^\infty  {f({t_k}){{(1 - s)}^{k - 1}}},
\end{align}
where the series converges  for all the points  $s \in  \mathbb{C} $.
\end{defn}

\begin{defn}
\cite{jonnalagadda2018matrix}
The Mittag-Leffler function for matrix arguments is 
$${F_\alpha }(B,{(t - a)^{\overline \alpha  }}) = \sum\limits_{n = 0}^\infty  {\frac{{{{(t - a)}^{\overline {\alpha n} }}}}{{\Gamma (\alpha n + 1)}}{B^n}},$$
where spectral radius of the matrix B is less than 1, $t \in N_a$,  $0 < \alpha  < 1$ and $B \in {R^{n \times n}}$.
\end{defn}

\begin{defn}
\cite{jonnalagadda2018matrix}
In fractional nabla calculus, the matrix exponential function for a matrix  $B \in {R^{n \times n}}$ is 
$${{{ e  }_{\alpha ,\alpha }}(B,{{(t - a)}^{\overline \alpha  }}) = \sum\limits_{}^{} {{{(t - a + 1)}^{\overline {\alpha (n + 1)} }}{B^n}(I - B)} }, $$
where spectral radius of the matrix B is less than 1, $t \in N_a$ and $0 < \alpha  < 1$.
\end{defn}

\begin{lemma}\label{lemma1a}
\cite {edwards2002boundedness}
Asymptotic stability of the linear system ${\nabla  }x(t) = Bx(t)$  is guaranteed if and only if every root of  the characteristic equation, $p(\lambda)= \lambda^2 - b \lambda + c$  is contained within the unit disk.
 \end{lemma}

\begin{lemma}\label{lemma1}
\cite {vcermak2021problem}
The linear system  ${}_a^C{\nabla ^\beta }x(t) = Bx(t)$ exhibits asymptotic stability if and only if the set of eigenvalues of matrix B  is confined to
\begin{equation}\label{9}
    {S_\beta } = \{ \arg(w) \in \mathbb{C}:\left| {\arg (\arg(w))} \right| > \frac{{\beta \pi }}{2} ~or~ \left| \arg(w) \right| > {(2\cos (\arg (\arg(w))/\beta ))^\beta }\},
    \end{equation}
where  $x(t)  \in {R^n}$, $B \in R^{n \times n}$ and $0 < \beta  < 1$.
 \end{lemma}

\begin{lemma}\label{lemma2}
\cite {acar2013exponential}
The following statement is valid
\begin{equation}\label{10}
 {e_{\alpha ,\alpha }}({t^{\overline \alpha  }}) = {(t + 1)^{\overline {\alpha  - 1} }}{F_{\alpha ,\alpha }}(\lambda ,{(t + \alpha )^{\overline \alpha}}).
 \end{equation}
\end{lemma}

\begin{lemma}\label{lemma3}
\cite{nechvatal2014asymptotics} Let $0 < \alpha  < 1$ and  $\lambda  \in {S_\alpha }$. Then, the discrete Mittag-Leffler function exhibits the subsequent property\\
\begin{equation}
F_{\alpha } (\lambda, {t_n}) = O({n^{ - \alpha }}) ~~ as ~~  n  \to \infty .
\end{equation}

\end{lemma}

\begin{lemma}\label{lemma4}
\cite{vcermak2021problem}
If any of the following conditions hold, the roots of  $p(\lambda)= \lambda^2 - b \lambda + c$ belong to ${S_\alpha }$ 
\begin{itemize}
\item[i].  $c > {4^\alpha } ~~  and  ~~b < {2^{ -\alpha }}c + 2^{ \alpha },$
\item[ii]. $0 < c \le {4^\alpha }   ~ ~  and  ~~   c < 2{\sqrt c }\cos \left( {\alpha \arccos \left( {\frac{1}{2}{\sqrt {c}}} \right)} \right),$
\item[iii].$c < 0    ~~  and  ~~ b > {2^{ -\alpha }}c + 2^{ \alpha },$
\end{itemize}
where  $0 < \alpha  < 1$.
\end{lemma}
\begin{thm}
The fractional order system
\begin{align}\label{eq:nD_theory}
{}_a^C{\nabla ^\alpha }X(t) = B\textsc{x}(t) + \mathbb{H}\left[\textsc{x}(t) \right],
\end{align}
exhibit asymptotic stability in the vicinity of its equilibrium point  $\textsc{x}^\star$, given the conditions
$${\lambda _i} \in {S_\alpha } ~~ and ~~\mathop {\lim }\limits_{\textsc{x} \to \textsc{x}^\star} \frac{||~ {\mathbb{H}\left[\textsc{x}(t)\right]}~ ||}{{\left\| \textsc{x}(t)-\textsc{x}^\star(t) \right\|}} = 0,$$ where   $\mathbb{H}\left[\textsc{x}(t) \right]$ is a  nonlinear term,  $\lambda_i$ are the eigenvalues of matrix $B \in {R^{n \times n}}$ and $\textsc{x}(t) \in {R^{n \times 1}}$. \end{thm}\label{th:nD_stability}

\begin{proof}
According to the method proposed in \cite{jonnalagadda2018matrix}, the unique solution of (\ref{eq:nD_theory}) is expressed as\\
$$\textsc{x}(t) = {F_\alpha }(B,{t^{\overline \alpha  }})\textsc{x}({t_0}) + \sum\limits_{s = 1}^t {e_{\alpha,\alpha}}(B,{{(t - s)}^{\overline \alpha  }})    {(I - B)^{ - 1}}H(\textsc{x}(t)).$$\\
Suppose ~$\textsc{x}^\star$ is the solution of system \eqref{eq:nD_theory}, Hence For every $\varepsilon>0$, there exists $\delta_0 >0$ such that if
$\left\| \textsc{x}(t)-\textsc{x}^\star \right\|< \delta_0 $ then $||~ {\mathbb{H}\left[\textsc{x}(t)\right]}~ || < \varepsilon \left\| \textsc{x}(t)-\textsc{x}^\star \right\| $. Thus, we  have
$$\left\| \textsc{x}(t)-\textsc{x}^\star \right\| \le \left\| {{F_\alpha }(B,{t^{\overline \alpha  }})} \right\| \left\| {\textsc{x}({t_0})} \right\|+ \sum\limits_{s = 1}^t \varepsilon {\left\| {\mathop {{e_{\alpha ,\alpha }}(B,{{(t - s)}^{\overline \alpha  }}}  } \right\|} {\left\| {I - B} \right\|^{ - 1}}{{\left\| \textsc{x}(t)-\textsc{x}^\star \right\|}}.$$\\
Given that $\delta$ is chosen arbitrarily such that $0 < \delta  < {\delta _0}$, and considering  lemma \ref{lemma2} we can derive
$$\left\| \textsc{x}(t)-\textsc{x}^\star  \right\| \le \delta \left\|  {{F_\alpha }(B,{t^{\overline \alpha  }})} \right\| +  M \sum\limits_{s = 1}^t {{{(t + 1)}^{\overline {\alpha  - 1} }}\left\| {{F_\alpha }(B,{{(t + \alpha  - s)}^{\overline \alpha  }}} \right\|},$$
\\where ~$M= \varepsilon \delta_0 {{(1 + \left\| B \right\|)^{-1}}} $,  regarding
${(t + 1)^{\overline {\alpha  - 1} }} = \frac{{{t^{\overline \alpha  }}}}{{t + 1}}$ proves that
$ \frac{{{t^{\overline \alpha  }}}}{{t + 1}} = O({t^{\alpha  - 1}})$,
 by lemma \ref{lemma3} we  have
$$\left\| \textsc{x}(t)-\textsc{x}^\star \right\| \le {\ell _1}{t^{ - \alpha }} + {\ell _2}{t^{\alpha  - 1}}\sum\limits_{s = 1}^t {{{(t)}^{ - 1 - \alpha }}} ,$$
$$\left\| \textsc{x}(t)-\textsc{x}^\star \right\| \le {\ell _1}{t^{ - \alpha }} + {\ell _2}{\ell _3}{t^{ - 1}},$$
$$\left\| \textsc{x}(t)-\textsc{x}^\star \right\| \to 0 ~~~  as ~~~ t\to \infty.$$
Asymptotic stability of the solution of system  (\ref{eq:nD_theory}) is suggested by this result.
\end{proof}
\begin{rem}
Assume $\mathbf{\Gamma}=bcirc(\mathbf{T},\tau,...,\tau)\in {\mathbb{C}^{N \times N}} $then
$$\lambda({\mathbf{\Gamma}})= \left\{  \lambda (\mathbf{T} + (N - 1)\tau),\underbrace {\lambda (\mathbf{T} - \tau),...,\lambda (\mathbf{T} - \tau)}_{N - 1} \right\}.$$
\end{rem}
\subsubsection{Analysis of the \textsc{FDMCRN} Model’s Stability}
 \begin{thm} \label{th:2D_stability}
The system  \eqref{eq:fmhnn_2_N} exhibits asymptotic stability at its equilibrium point,  ${E^ * } = (\sigma ,\sigma  - \frac{a}{{{\sigma ^2} + 1}},\sigma ,\sigma  - \frac{a}{{{\sigma ^2} + 1}},{\gamma _1},{\gamma _2})$ where ${\gamma _1},{\gamma _2}$ are constants representing specific positions on the line of equilibrium points, if and only if, in each case, any of the following conditions holds true\\
\textbf{Case1}:
\begin{itemize}
\item[i].  $c_0 > {4^\alpha },  ~~   2^{\alpha} - 2^{-\alpha} \eta  <  \vartheta  - k {\gamma _1},$
\item[ii]. $0 < c_0 \le {4^\alpha },   ~~    - 2~ {\eta^{\frac{1}{2}}} cos(\frac{\alpha}{2} cos^{-1}({\eta^{{\frac{1}{{2\alpha }}}}}))  <   \vartheta  -  k {\gamma _1},$
\item[iii]. $c_0 < 0,  ~~ 2^\alpha - 2^{-\alpha} \eta >  \vartheta  -  k {\gamma _1},$
\end{itemize}
\textbf{Case2}:
\begin{itemize}
\item[i].  $b_0 > {4^\alpha }, ~~  \vartheta  + k{\gamma _1} + 2k{\gamma _2} <   2^{1-\alpha}(k{\gamma _1}{\gamma _2}^2 - \vartheta k{\gamma _2}) - {2^\alpha },$
\item[ii]. $0 < b_0 \le {4^\alpha }, ~~ \vartheta  + k{\gamma _1} + 2k{\gamma _2} <  2cos[\frac{\alpha}{2} co{s^{ - 1}}[{(2{k^2}{\gamma _1}{\gamma _2} + 2\vartheta k{\gamma _2})^{\frac{1}{{2\alpha }}}})]]{(2{k^2}{\gamma _1}{\gamma _2} + 2\vartheta k{\gamma _2})^{\frac{1}{2}}},$
\item[iii]. $b_0 < 0, ~~ \vartheta  + k{\gamma _1} + 2k{\gamma _2} <   2^{1-\alpha}(k{\gamma _1}{\gamma _2}^2 - \vartheta k{\gamma _2}) - {2^\alpha },$
\end{itemize}
where  $0 < \alpha  \le 1$, $\vartheta=\frac{{k{\gamma _1} - 2a\sigma }}{{{{({\sigma ^2} + 1)}^2}}} - 1$ and  ${b_i},{c_i}, (i = 0 ,1)$   are listed in Table \ref{t1}.
\end{thm}
\begin{proof} The system \eqref{eq:fmhnn_2_N} has the following Jacobian matrix

\begin{align}\label{eq2}
\left[\begin{array}{*{20}{c}}
{\frac{{k{\varphi _1} - 2a{x_1}}}{{{{({x_1}^2 + 1)}^2}}} - 1}&1&{ - k{\varphi _1}}&0&{k({x_1} - {x_2})}&0\\
{ - \eta }&{k{\varphi _2}}&0&{ - k{\varphi _2}}&{k({y_1} - {y_2})}&0\\
{ - k{\varphi _1}}&1&{\frac{{k{\varphi _1} - 2a{x_1}}}{{{{({x_1}^2 + 1)}^2}}} - 1}&0&{ - k({x_1} - {x_2})}&0\\
0&{ - k{\varphi _2}}&{ - \eta }&{k{\varphi _2}}&{ - k({y_1} - {y_2})}&0\\
1&0&{ - 1}&0&0&0\\
0&1&0&{ - 1}&0&0
\end{array}\right],
\end{align}
which at its equilibrium,  will be
\begin{align}\label{genera}
J(E^\star) =
\left[\begin{array}{*{20}{c}}
{\frac{{k{\gamma _1} - 2a\sigma }}{{{{({\sigma ^2} + 1)}^2}}} - 1}&1&{ - k{\gamma _1}}&0&0&0\\
{ - \eta }&{k{\gamma _1}}&0&{ - k{\gamma _1}}&0&0\\
{ - k{\gamma _1}}&1&{\frac{{k{\gamma _1} - 2a\sigma }}{{{{({\sigma ^2} + 1)}^2}}} - 1}&0&0&0\\
0&{ - k{\lambda _1}}&{ - \eta }&{k{\gamma _1}}&0&0\\
1&0&{ - 1}&0&0&0\\
0&1&0&{ - 1}&0&0
\end{array}\right].
\end{align}
Assuming ${a_i}~ $(i = 0 to 3 ) and  ${b_i},{c_i}$ (i = 0, 1 ) presented in Tables \ref{t1}, \ref{t1b}, the Jacobian matrix's characteristic equation is represented  as $p(\lambda ) = \lambda^2p_1(\lambda)$
where
\begin{align}\label{eq:p1}
p_1(\lambda)={\lambda ^4}
+ a_3{\lambda ^3}  + a_2{\lambda ^2} + a_1\lambda + a_0,
\end{align}
by splitting the quadratic polynomial by factoring it into two quadratics we can rewrite $ p_1(\lambda) $ in the form of
\begin{align}\label{eq:p2}
p_1(\lambda)=({\lambda ^2} +b_1\lambda + b_0)({\lambda ^2}  + c_1\lambda + c_0)=p_2(\lambda) \times p_3(\lambda).
\end{align}
According to lemma \eqref{lemma4} the proof follows directly.\\
\end{proof}


\begin{table}[!ht]
\begin{center}
\caption{The coefficients of  the characteristic polynomial $p_1(\lambda)$ in eq. (\ref{eq:p1}).}
\label{t1}
\begin{tabular}{l|cccc}
\hline
\multirow{2}{0.9cm}{coefficients}
\multirow{2}{0.7cm}
{~~~~~~~~~~~~~~~~~~~~~~~~~~~~~~~~~~~~~~~~~~~~$\textbf{p}_1$($\lambda$)}&
\\
\\
\hline
$a_0$ & $2({\gamma _1}{\gamma _2}\eta{k^2} + {\gamma _2}\eta\nu k)$  \\
$a_1$ & $-(- 2{\gamma _2}{\gamma _1}^2{k^3} + \eta{\gamma _1}k + 2{\gamma _2}k{\nu ^2} +
2{\gamma _2}\eta k + \eta \nu)$ \\
$a_2$ & $- {\gamma _1}^2{k^2} +4{\gamma _2}k\nu + {\nu^2} + \eta$\\
$a_3 $& $-(2\nu + 2{\gamma _2}k) $ \\
\hline
\end{tabular}
\end{center}
\end{table}
\begin{table}[!ht]
\begin{center}
\caption{The coefficients of  the characteristic polynomial $p_2(\lambda)$, $p_3(\lambda)$ in eq. (\ref{eq:p2}).}
\label{t1b}
\begin{tabular}{l|cccc}
\hline
\multirow{2}{0.9cm}{coefficients}
\multirow{2}{0.7cm}{~~~~~~~~~~~~~~~~~$\textbf{p}_2$($\lambda$)}&
\multirow{2}{0.7cm}{~~~~~~~~~~~~~~~~~~~~~~~~~~~~$\textbf{p}_3$($\lambda$)}&
\\
\\
\hline
\rowcolor{gray!20}
$b_0$ &  $ 2{\gamma _1}{\gamma _2}{k^2} + 2{\gamma _2}\vartheta k$ & - \\
\rowcolor{gray!20}
$b_1$ &  $ - \vartheta - {\gamma _1}k - 2{\gamma _2}k$ & - \\
$c_0$ & - &$\eta $  \\
$c_1$ & - &$\gamma _1 k - \vartheta$ \\
\hline
\end{tabular}
\end{center}
\end{table}

\begin{rem}
The roots of $p_i(\lambda), (i=1,2,3)$ must be located in the region ${S_\alpha }$ which is shown in Fig.~\ref{fig3}. Note that the Matignon sector \cite{matignon1998generalized}, which is located left to the dashed half-lines is a subset of ${S_\alpha }$.
\end{rem}
\begin{figure}[!tbp]
 \begin{center}
\includegraphics[width=12cm, height=5cm]{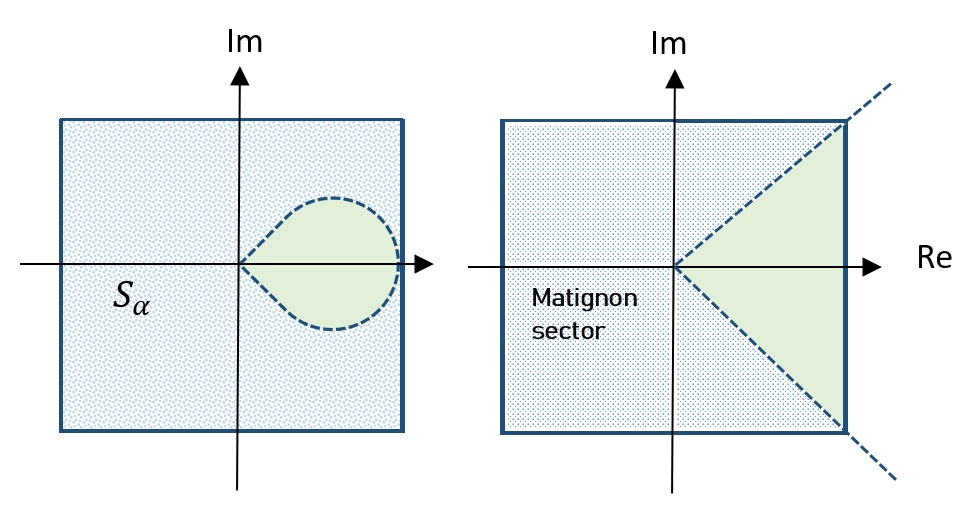}
\end{center}
\caption{Asymptotic stability regions,  ${S_\alpha }$ and Matignon sector.}\label{fig3}
\end{figure}

 \subsubsection{ Analysis of the  \textsc{FDMCRN} Model’s Stability with a Ring-based Structure. }
\begin{thm} \label{th:stability_nSub}
The system  \eqref{ndim} exhibits asymptotic stability at its equilibrium point,  ${E^ * } = (\boldsymbol{\sigma},\boldsymbol{\sigma}^*,\boldsymbol{\sigma} ,\boldsymbol{\sigma}^*,\boldsymbol{\gamma _1},\boldsymbol{\gamma _2})$, if and only if, in each case, any of the following conditions holds true\\
\\
\\
\\
\textbf{Case1}:
\begin{itemize}
\item[i].  $N - 1 > {4^\alpha }, - {\rm{(3D  +  2P  -  DN)/(2P)}} < {2^{ - \alpha }}(N - 1) + {2^\alpha },$
\item[ii]. $0 < N - 1 < {4^\alpha },  - (3D + 2P - DN)/(2P) < \cos (\alpha {\cos ^{ - 1}}\frac{{{{(N - 1)}^{\frac{1}{{2\alpha }}}}}}{2}),$
\end{itemize}
\textbf{Case2}:
\begin{itemize}
\item[i].  $\eta{(N - 1)^2} > {4^\alpha }, - (\eta{(N - 1)^2}) < {2^{ - \alpha }}(\eta{(N - 1)^2}) + {2^\alpha },$
\item[ii]. $ 0 < \eta(N - 1)^2 < {4^\alpha }, - (\eta {(N - 1)^2})/(2P) < \cos (\alpha {\cos ^{ - 1}}\frac{{{{(N - 1)}^{\frac{1}{{2\alpha }}}}}}{2}), $
\end{itemize}
\textbf{Case3}:
\begin{itemize}
\item[i].  $N - 1 > {4^\alpha },{\rm{(3D  +  2P  -  3DN  -  2NP)/(2P)}} < {2^{ - \alpha }}(N - 1) + {2^\alpha },$
\item[ii]. $0 < N - 1 < {4^\alpha }, (3D + 2P - 3DN - 2NP)/(2P) < \cos (\alpha {\cos ^{ - 1}}\frac{{{{(N - 1)}^{\frac{1}{{2\alpha }}}}}}{2}),$
\end{itemize}
where  $ \alpha  \in (0,1)$.
\end{thm}
\begin{proof}
The compact form of the network \eqref{ndim} is presented as $    \mathbb{F}(X)=\mathbb{A} X + \mathbb{H}(X) $,
matrix $A \in {\mathbb{R}^{N \times N}}$ is a block circulant matrix constructed from the blocks $ A_0,A_1,...,A_1 $, such that  $A_0,A_1 \in {\mathbb{R}^{6 \times 6}}$,
$\mathbb{H}(\textsc{X})=  \Psi (X)+ k \widetilde \Phi \textsc{X}$ $ ~\text{and} ~\widetilde \Phi =bcirc(\Phi ,0,..,0)_{1\times N}$  as

$$\mathbb{A}_0=
\left[\begin{array}{*{20}{c}}
{ - 1-D/P}&1&0&0&0&0\\
{ - \eta }&0&0&0&0&0\\
0&0&{ - 1}&1&0&0\\
0&0&{ - \eta }&0&0&0\\
1&0&{ - 1}&0&0&0\\
0&1&0&{ - 1}&0&0
\end{array}\right],$$ ~~
\\
$$ \mathbb{A}_1=
\left[\begin{array}{*{20}{c}}
{ D/2P}&0&0&0&0&0\\
0&0&0&0&0&0\\
0&0&0&0&0&0\\
0&0&0&0&0&0\\
0&0&0&0&0&0\\
0&0&0&0&0&0
\end{array}\right],$$ ~~
\\
$$\Psi (X_n)=
\left[\begin{array}{*{20}{c}}
{\frac{a}{{1 + x_{1,n}^2}}}\\
{\eta \sigma }\\
{\frac{a}{{1 + x_{2,n}^2}}}\\
{\eta \sigma }\\
0\\
0
\end{array}\right], ~~
\Phi=
\left[\begin{array}{*{20}{c}}
{{\varphi _1}}&0&{ - {\varphi _1}}&0&0&0\\
0&{{\varphi _2}}&0&{ - {\varphi _2}}&0&0\\
{ - {\varphi _1}}&0&{{\varphi _1}}&0&0&0\\
0&{ - {\varphi _2}}&0&{{\varphi _2}}&0&0\\
0&0&0&0&0&0\\
0&0&0&0&0&0
\end{array}\right].$$
\\
Next, we prove that equation \eqref{ndim} meets the stability condition at its equilibrium point.
We have
\begin{equation*}
    \mathop {\lim }\limits_{X \to X^*} \frac{||\mathbb{H}(X)||}{||X -X^\star||}
    =
    \mathop {\lim }\limits_{X \to X^*} \frac{ k\left\| {\widetilde \Phi } \right\|\left\| X \right\|}{|| X -X^\star||}
    = 0,
\end{equation*}
where
\begin{align*}
 p(\lambda) = \det( \mathbb{A}-\lambda I)
    &=\left\{ p_{1}(\lambda) \times p_2({\lambda)} \right\},\\
\end{align*}
\begin{align}\label{eq:3}
{p_1}(\lambda ) =\lambda^2 ({\lambda^4} + a_{3}{\lambda^3} + a_{2}{\lambda^2} +a_{1} \lambda^{} + a_{0}),
\end{align}

\begin{align}\label{eq:4}
{p_2}(\lambda ) =\lambda^2 ( {\lambda^4} + b_{3}{\lambda^3} + b_{2}{\lambda^2} +b_{1} \lambda^{} + b_{0}).
\end{align}
The coefficients of  the characteristic polynomial $p_1(\lambda)$ and $p_2(\lambda)$ are presented in Table \ref {t4} and Table \ref {t5}. By factoring each of the quartic polynomials into two quadratic factors, we can rewrite
\begin{align*}
{p_1}(\lambda ) =\lambda^2({\lambda^2} +s_{1} \lambda^{} + s_{0}) ({\lambda^2} +u_{1} \lambda^{} + u_{0}),\\
{p_2}(\lambda ) = \lambda^2({\lambda^2} +q_{1} \lambda^{} + q_{0}) ({\lambda^2} +v_{1} \lambda^{} + v_{0}),\\&
\end{align*}
where\\
\\
$\left\{ \begin{array}{l}
{s_1} = {\rm{( - 3D  -  2P  + 3DN  +  2NP)/(2P)}},\\
{s_0} = N - 1,\\
{u_1} = \eta{N^2} - 2 N \eta+ \eta,\\
{u_0} = \eta{N^2} - 2 N \eta+ \eta,
\end{array} \right.$
$\left\{ \begin{array}{l}
{q_1} = {\rm{(3D  +  2P  -  DN)/(2P)}},\\
{q_0} = N - 1,\\
{v_1} = \eta{N^2} - 2 N \eta+ \eta,\\
{v_0} = \eta{N^2} - 2 N \eta+\eta.
\end{array} \right.$
\\
\\
\\
By applying Lemma  \eqref{lemma4}, the proof becomes straightforward.
\end{proof}

\begin{table}[!ht]
\begin{center}
\caption{The coefficients of  the characteristic polynomial $p_1(\lambda)$ in eq. (\ref{eq:3}).}
\label{t4}
\begin{tabular}{l|cccc}
\hline
\multirow{2}{0.9cm}{coefficients}
\multirow{2}{0.7cm} {$~~~~~~~~~~~~~~~~~~~~~~~~~~~~~~~~~~~~~~~~~~~~\textbf{p}_1$($\lambda$)}&
\\
\\
\hline
$a_0$ & $(2p {\eta}^2 N^4 - 8p {\eta}^2 N^3 + 12p {\eta}^2 N^2- 8p {\eta}^2 N+ 2p {\eta}^2)/(2p)$  \\
$a_1$ & $- (3d {\eta}  + 4 {\eta} p + 12{\eta} N^2p - 4{\eta} N^3p - 9d{\eta} N - 12{\eta} Np + 9d{\eta} N^2 - 3d{\eta} N^3)/(2p)$ \\
$a_2$ & $(3d + 2p - 6dN + 4{\eta} p - 4Np + 3dN^2 + 2N^2p + 4{\eta} N^2p - 8{\eta} Np)/(2p) $\\
$a_3$ & $- (3d + 4p - 3dN - 4Np)/(2p) $ \\
\hline
\end{tabular}
\end{center}
\end{table}
\begin{table}[!ht]
\begin{center}
\caption{The coefficients of  the characteristic polynomial $p_2(\lambda)$ in eq. (\ref{eq:4}).}
\label{t5}
\begin{tabular}{l|cccc}
\hline
\multirow{2}{0.9cm}{coefficients}
\multirow{2}{0.7cm} {$~~~~~~~~~~~~~~~~~~~~~~~~~~~~~~~~~~~~~~~~~~~~\textbf{p}_2$($\lambda$)}&
\\
\\
\hline
$b_0$ & $(2p{\eta}^2 N^4 - 8p{\eta}^2 N^3+ 12p {\eta}^2 N^2 - 8p{\eta}^2 N + 2p{\eta}^2)/(2p)$  \\
$b_1$ & $(3d{\eta} - 4{\eta} N^2p + 2{\eta} N^3p - 7d{\eta} N + 2{\eta} Np + 5d{\eta} N^2 - d{\eta} N^3)/(2p) $ \\
$b_2$ & $- (3d + 2p - 4dN - 4{\eta}ap - 2Np + dN^2 - 4{\eta} N^2p + 8{\eta} Np)/(2p $\\
$b_3$ & $(3d - dN + 2Np)/(2p)  $ \\
\hline
\end{tabular}
\end{center}
\end{table}


\section{Numerical Simulation}
This section presents two numerical examples to illustrate and validate the theoretical stability results. The numerical simulations presented in this section were conducted using MATLAB.

\subsubsection*{\textbf{Particular Scenario: Example I}}
Similar to the approach used by Caixia Wang \textit{et al}.~\cite{wang2015stability}, we set $a =3,~ k=0.05,~ \mu=0.001~ and  ~\sigma =-1.5$ in system (\ref{eq:fmhnn_2_N}). Consequently, the  equation  \eqref{eq:p1} simplifies to
\begin{align}\label{eq:p1_lambda2}
p_1(\lambda)=({\lambda ^2} +b_1\lambda + b_0)({\lambda ^2}  + c_1\lambda + c_0).
\end{align}
Through assumption of $~b_1=\frac{{{\rm{[25 -  (185}}{\gamma _1}{\rm{k)]}}}}{{169}}{\rm{ -  2}}{\gamma _2}{k},~ b_0=2{\gamma _2}[\frac{k}{{169}}(16{\gamma _1}k - 25) + {\gamma _1}{k^2}]$, $c_1 =\frac{{{\rm{ [153}}{\gamma _1}{\rm{k +  25]}}}}{{169}}
~ and ~ c_0= 0.001 $ in Theorem (\ref{th:2D_stability}),  the stability criteria of model (\ref{eq:fmhnn_2_N}) are satisfied  if and only if
\begin{align*}
0 < {c_0} \le {4^\alpha },~~ - {c_1} < 2\sqrt {{c_0}} \cos (\alpha {\cos ^{ - 1}}\frac{{{{({c_0})}^{\frac{1}{{2\alpha }}}}}}{2}),
\end{align*}
and provided that any of these qualifications is met\\
\begin{itemize}
\item[i]. ${b_0} > {4^\alpha },~~{b_1} >  - {2^{ - \alpha }}{b_0} - {2^\alpha },$
\item[ii]. $ 0 < {b_0} \le {4^\alpha }, ~~- {b_1} < 2\sqrt {{b_0}} {\cos }(\alpha \cos^{ - 1} (\frac{{{{({b_0})}^{\frac{1}{{2\alpha }}}}}}{2})),$
\item[iii]. ${b_0} < 0,~~{b_1} \prec  - {2^{ - \alpha }}{b_0} - {2^\alpha }.$
\end{itemize}
The above said stability conditions can be expressed in terms of upper and lower bounds of $\gamma _1$ and $\gamma _2$ values. Then,  the stability criteria of model (\ref{eq:fmhnn_2_N}) is met  if and only if,\\
\begin{equation}\label{eq-gamma1}
 {\gamma _1} > \frac{{169}}{{153k}}( - \frac{{25}}{{169}} - 2\sqrt {0.001} \cos (\alpha {\cos ^{ - 1}}\frac{{{{(0.001)}^{\frac{1}{{2\alpha }}}}}}{2})),
\end{equation}
and provided that any of these qualifications is met
\begin{itemize}
\item[1]. $
2k{\gamma _2}{R_1} > {4^\alpha }, ~
{k\gamma _2}[2- R_1 {2^{ - \alpha  + 1}}]<  {2^\alpha } -{R_1},
$
\item[2]. $
0<2k{\gamma _2}{R_1} < {4^\alpha },~
{2^{ - \alpha }}k{\gamma _2} + {\alpha ^2}{(\frac{\pi }{2} - \frac{{{{(2k{\gamma _2}{R_1})}^{\frac{1}{{2\alpha }}}}}}{2})^2}<1- R_1 {2^{-\alpha  - 1}},
$
\item[3]. $
2k{\gamma _2}{R_1}< 0, ~
  {2^\alpha } - {R_1}  < {k\gamma _2}[{2-R_12^{ - \alpha  + 1}}],
$
\end {itemize}
where $ R_0 =   - \frac{{25}}{{169}} - 2\sqrt {0.001} \cos (\alpha {\cos ^{ - 1}}\frac{{{{(0.001)}^{\frac{1}{{2\alpha }}}}}}{2})$~ and ~ $ R_1= \frac{185}{{153}}{R_0} - \frac{{25}}{{169}} .$
\begin{table}
\begin{center}
\begin{tabular}{ |c||c|c|c| }
\hline
\multicolumn{4}{|c|}{Stability conditions based on upper and lower boundary of $\gamma_2$} \\
\hline
Fractional order & Condition (1) & Condition (2) &  Condition (3) \\
\hline
$\alpha=0.99$ & $\gamma _2<$ -119.9058  &  -4106.6 $ <$ $\gamma_2 <$  -4.8162 & $\gamma _2>$0\\
\hline
\end{tabular}
\caption{Stability region of two neuron  model \ref{eq:fmhnn_2_N} of  Example I .}\label{t3}
\end{center}
 \end{table}\\
 \\
The model exhibits a line equilibrium point, which is represented as
$$E=(\sigma, \sigma -\frac{a}{{{\sigma ^2} + 1}}, \sigma, \sigma -\frac{a}{{{\sigma ^2} + 1}}, \gamma_1, \gamma_2), $$
where $\gamma_1~ and ~ \gamma_2$  are constants that indicate specific positions along the line of equilibrium points. \\
Using the defined boundaries for  $\gamma_1$  and $\gamma_2$ in  Eq.~(\ref{eq-gamma1}) and Table \ref{t3}, the stability properties of the map are thoroughly examined and discussed across three distinct cases, highlighting the varying behaviors under which stability is maintained or disrupted.\\
\textbf{Case 1}; for $\alpha=0.99$,  with  $\gamma_1= -1$ and  $\gamma_2= -130$ , condition 1 in Table \ref{t3} is satisfied.  The eigenvalues of the  Jacobian matrix can then be determined in the set, $\Lambda$= $\{ -13,  -0.2027,  0,  0,  -0.0109,  -0.0918\}.$
Due to negativity of the eigenvalues, all the eigenvalues are located in $S_\alpha$ and the system is stable,  Fig.~\ref{fig4}.
\begin{figure}[!ht]
\centering
\includegraphics[width=14cm, height=8cm]{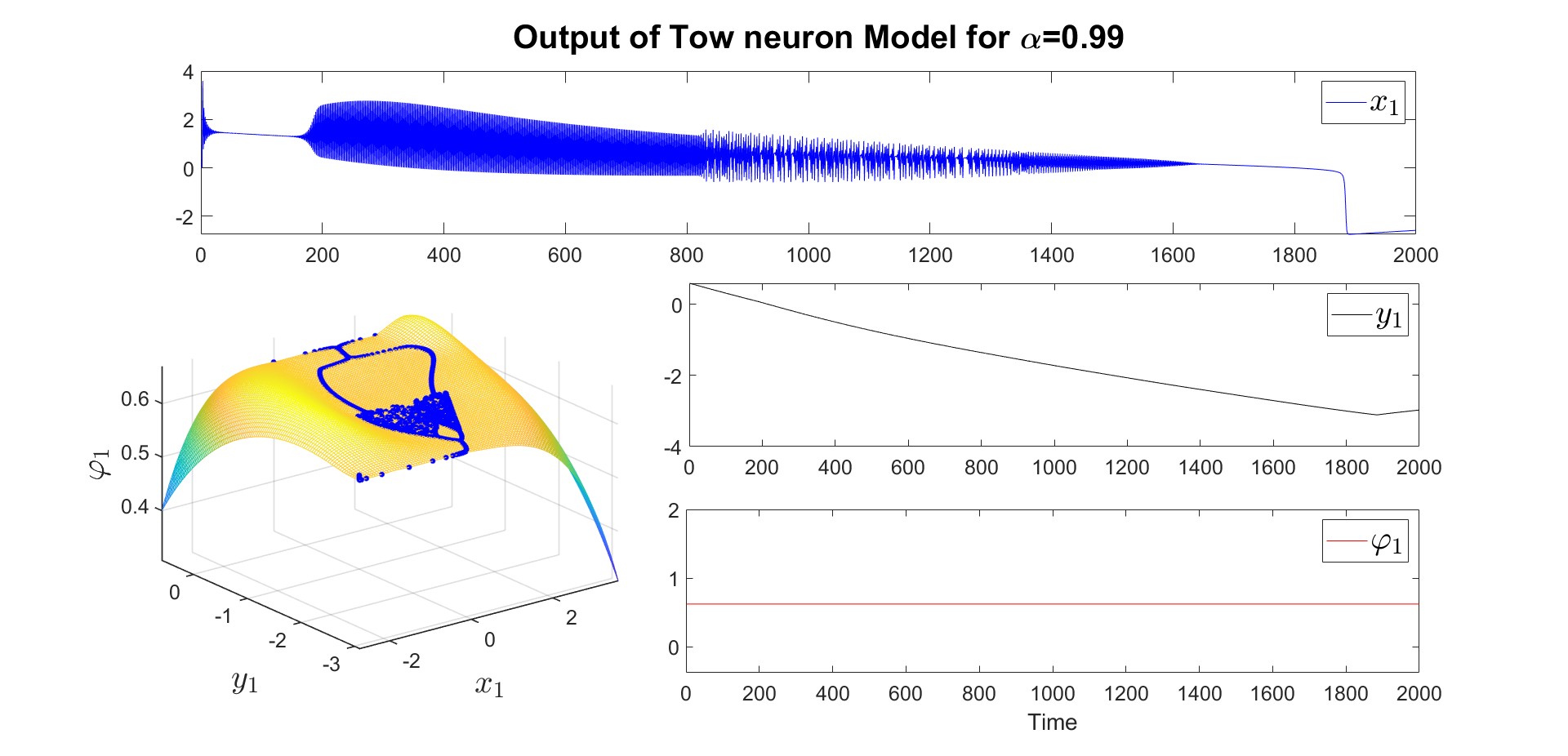}
\caption{The behavior over time of \textsc{FDMCRN} model (\ref{eq:fmhnn_2_N}) in Case 1. } \label{fig4}
\end{figure}
\\
\textbf{Case 2};  for $\alpha=0.99$,  with $\gamma_1= -2 $ and  $\gamma_2=  -45$, condition 2 in Table \ref{t3} is satisfied.  The eigenvalues of the  Jacobian matrix can subsequently be calculated in the  set $\Lambda$=$\{ -4.5,  -0.2574,  0,  0,  -0.0287 - 0.0133i,  -0.0287+ 0.0133i \}.$
All the eigenvalues have negative real part which approve the stability of the model based on their location in $S_\alpha$,   Fig.~\ref{fig5}.
\begin{figure}[!ht]
\centering
\includegraphics[width=14cm, height=8cm]{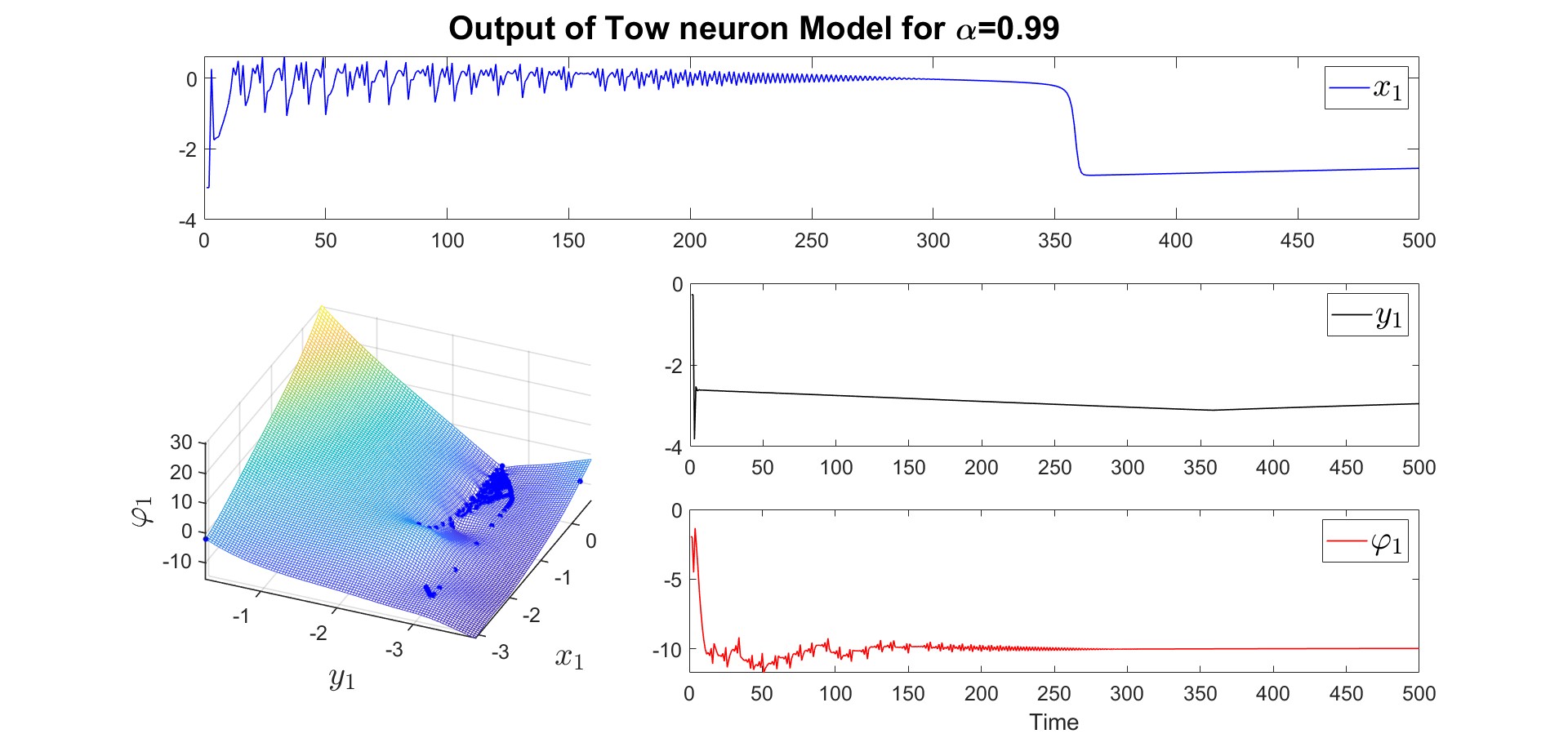}
\caption{The behavior over time of  \textsc{FDMCRN} model (\ref{eq:fmhnn_2_N})  in Case 2.} \label{fig5}
\end{figure}
\\
\textbf{Case 3};  for $\alpha=0.99$,  setting $\gamma_1= -3$  and  $\gamma_2=  50$, condition 3 in Table \ref{t3} is satisfied.The eigenvalues of  the Jacobian matrix  are in the  set, $\Lambda$=$\{-0.3121,-0.2574, 0, 0,\\ -0.0061 - 0.0310i,  -0.0061 + 0.0310i \}.$ All the eigenvalues are in $S_\alpha$ and stability is approved, Fig.~\ref{fig6}.
\begin{figure}[!ht]
\centering
\includegraphics[width=14cm, height=8cm]{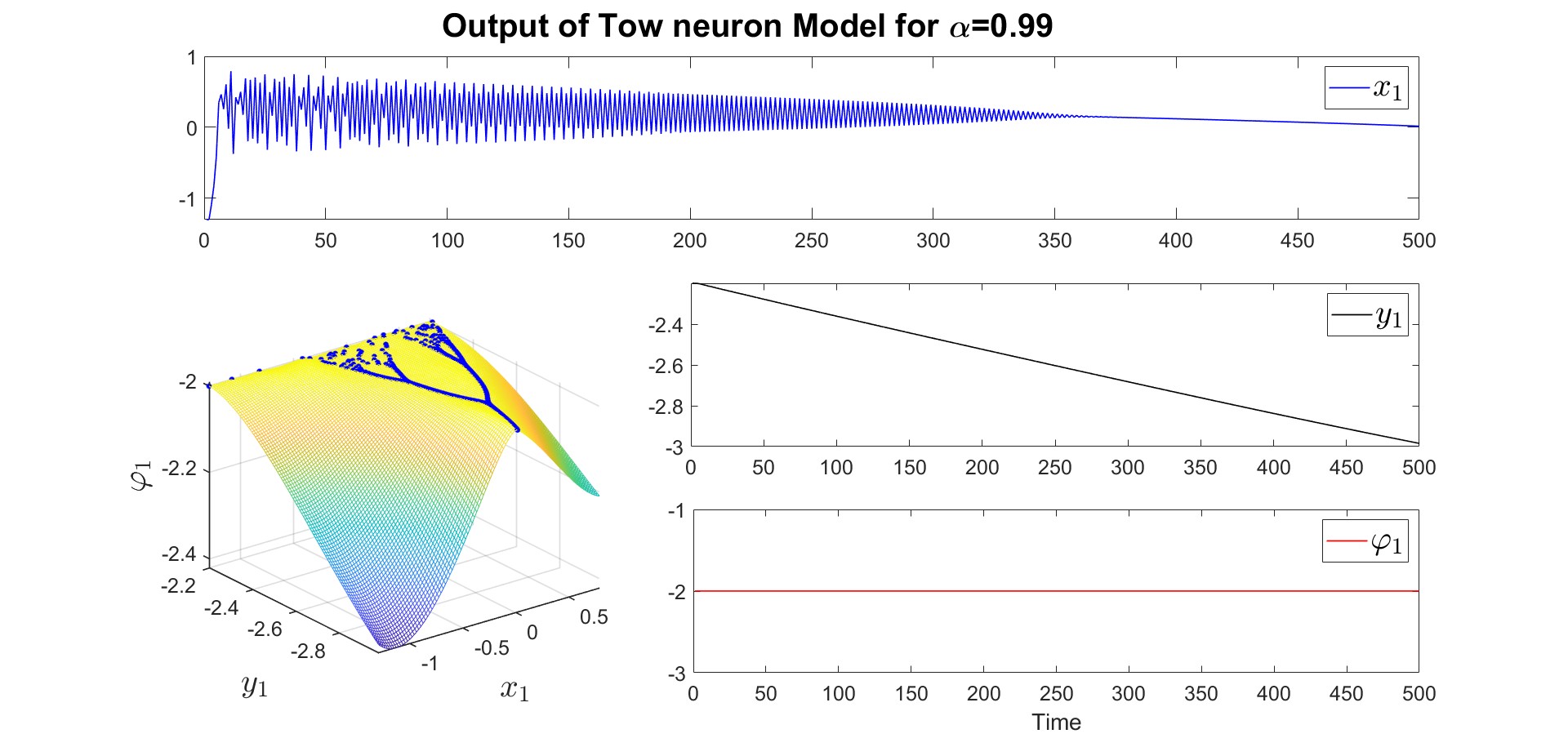}
\caption{The behavior over time of \textsc{FDMCRN} model (\ref{eq:fmhnn_2_N}) in Case 3. } \label{fig6}
\end{figure}
\\
Considering different values of $\alpha = 0.99, 0.89,$ and $0.79$, the time series simulation indicates that the time required for the system to stabilize depends on the derivative's order. As the order of derivation approaches one, the system reaches stability more quickly (Fig.~\ref{fig7}).\\
Specifically, when $\alpha=1$, Lemma~\ref{lemma1a} indicates that the eigenvalues associated with system (\ref{eq:fmhnn_2_N}) lie outside the unit disk, implying system instability, as illustrated in Fig.~\ref{fig8}.\\
\\
\\
\\
\\
\begin{figure}[!ht]
\centering
\includegraphics[width=14cm, height=8cm]{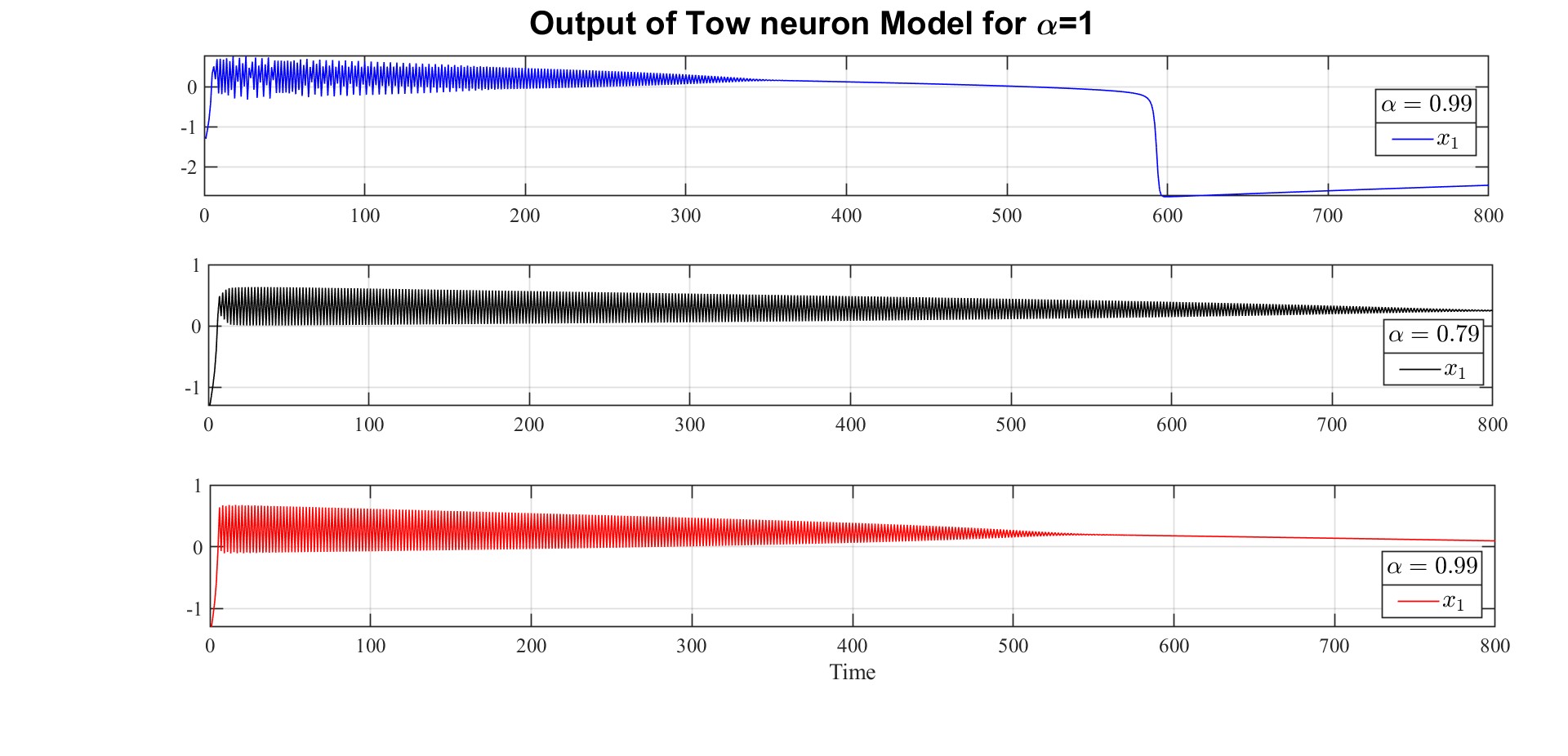}
\caption{The behavior over time of \textsc{FDMCRN} model (\ref{eq:fmhnn_2_N}) by  three different fractional orders. } \label{fig7}
\end{figure}

\begin{figure}[!ht]
\centering
\includegraphics[width=14cm, height=8cm]{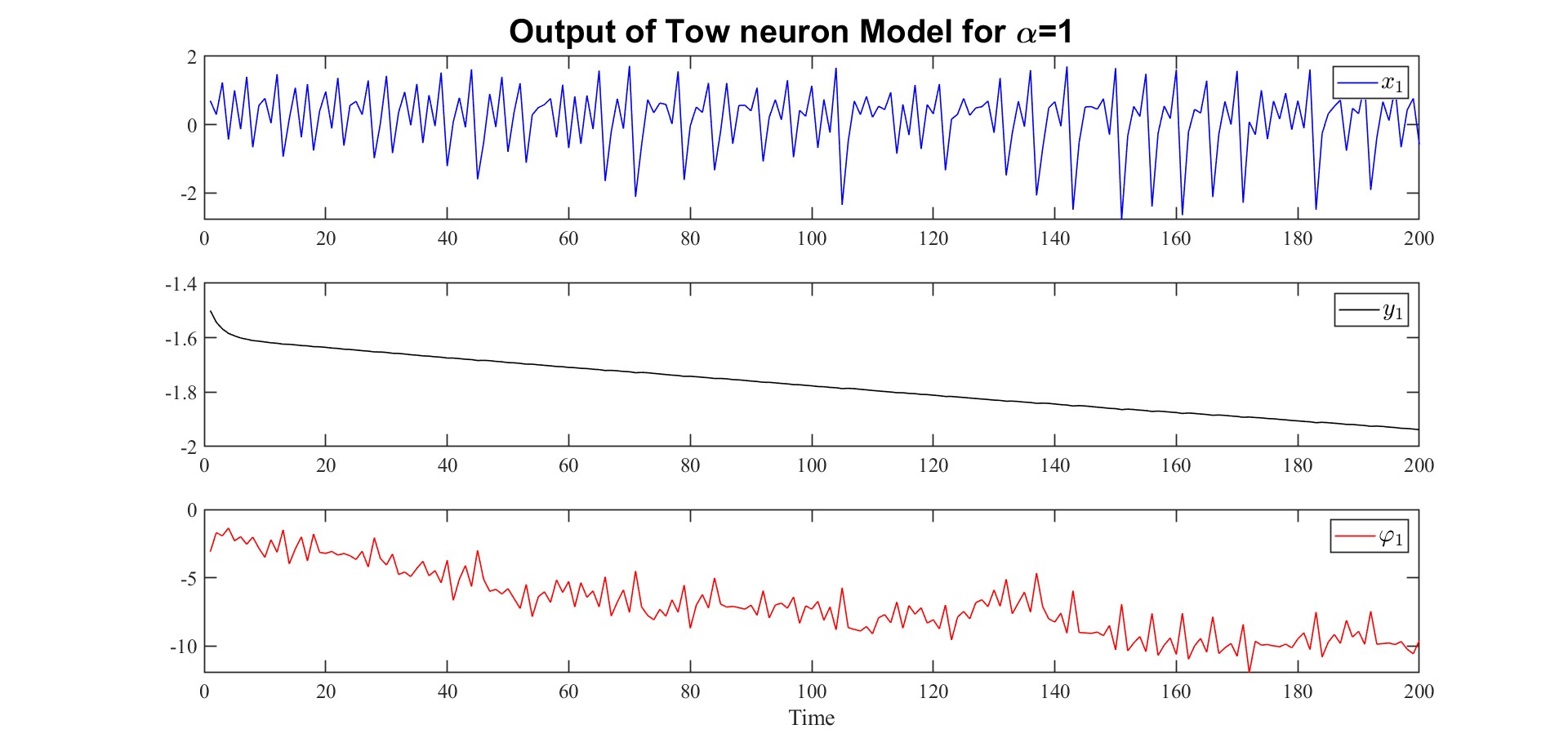}
\caption{The behavior over time of \textsc{FDMCRN} model (\ref{eq:fmhnn_2_N}) where $\alpha=1$. } \label{fig8}
\end{figure}

\subsubsection*{\textbf{Particular Scenario: Example II}}
In this section, alongside the parameter assumptions in the model  \eqref{ndim}, we define the parameters  $P$,  $D$, $N$, and $\alpha$.  The parameters are according to the conditions outlined in Theorem (\ref{th:stability_nSub}) and the  corresponding results are shown in
Table \ref{t6}. In the table,  (\cmark,  \xmark)  verify which stability cases are met and which are not. We examine two cases and demonstrate that the stability of the ring-structured model (\ref{ndim}) depends on several factors, including  $N$, $\alpha$, $D$, and $P$.
\begin{table}
\centering
\begin{tabular}{ |c||c|c|c| }
\hline
\multicolumn{4}{|c|}{Stability region for ring network  \eqref{ndim}} \\
\hline
Features & Case 1 & Case 2 &  Case 3 \\
\hline
$\alpha$=0.98, D=0.3, P=2, N=4 & i \cmark ~ ii \xmark  &  i \xmark ~ ii \cmark     & i \cmark ~ ii \xmark \\
\hline
$\alpha$=0.84, D=0.1, P=1, N=3 & i  \xmark ~ ii \cmark   & i \xmark ~ ii \cmark   & i \xmark ~ ii \cmark \\
\hline
\end{tabular}
\caption{Stability region of network model \eqref{ndim} of  Example II.}\label{t6}
 \end{table}
\\
For example, with the parameters $(P, D, N, \alpha) = (2, 0.3, 4, 0.98) $, we demonstrated that the model is asymptotically stable. This is verified by numerically solving the model, as shown in  Fig.~\ref{fig9} , where the chosen parameters satisfy the asymptotic stability conditions outlined in Theorem (\ref{th:stability_nSub}). We also conducted an additional simulation using the parameters $(P, D, N, \alpha) = (1, 0.1, 1, 0.84) $, as shown in  Fig.~\ref{fig10}. In this case as well, the model's parameter settings also satisfy the asymptotic stability conditions specified in Theorem (\ref{th:stability_nSub}).\\
\\
The time series analysis, for $P=2$, $D=0.1$, $N=2$, shown in Fig.~\ref{fig11}, reveals that  the $\alpha$ parameter significantly affects the time required for the system to reach stability. As $\alpha$ changes, the time taken for the system to stabilize also varies, highlighting the crucial role of this parameter in the system's dynamic behavior. As $\alpha$ decreases, the time required for the system to reach stability becomes longer, as shown in Fig.~\ref{fig11}.This highlights the inverse relationship between $\alpha$ and the stabilization time.

\begin{figure}[!ht]
\centering
\includegraphics[width=15cm, height=7cm]{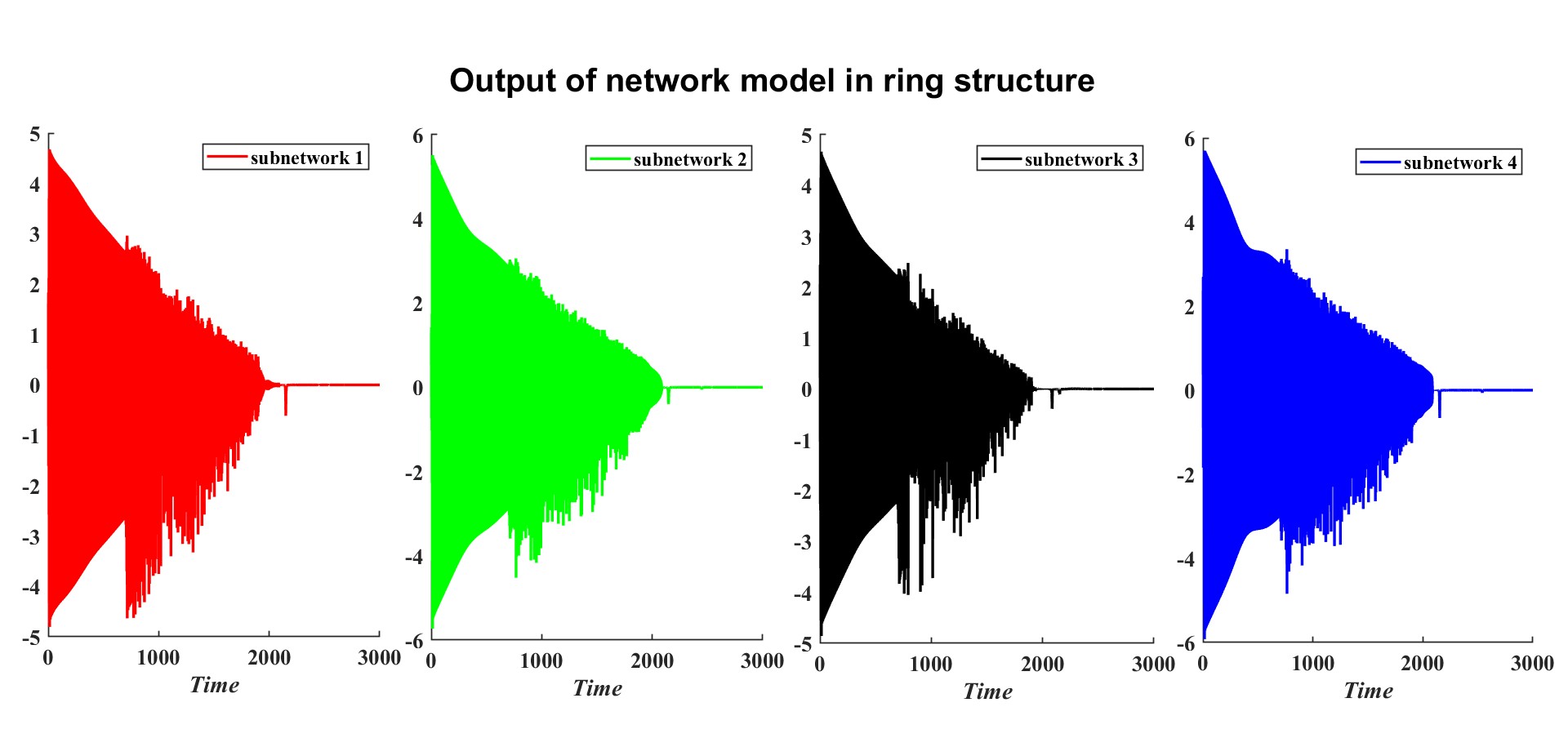}
\caption{The behavior over time of the model  (\ref{ndim}), respect to $X_1$,  with the first feature in Table \ref{t6}. } \label{fig9}
\end{figure}

\begin{figure}[!ht]
\centering
\includegraphics[width=16cm, height=7cm]{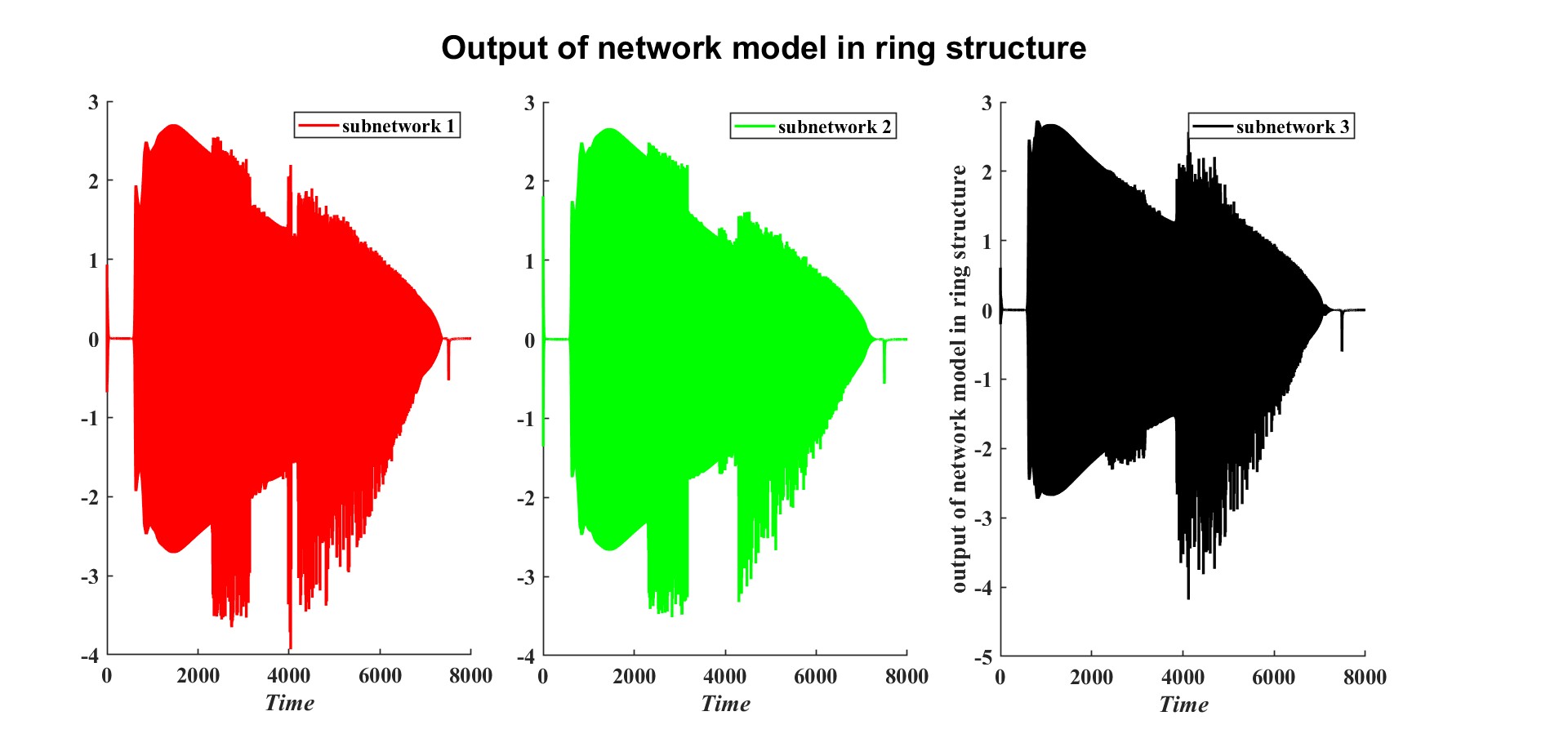}
\caption{The behavior over time of the model  (\ref{ndim}), respect to $X_1$,  with the second  feature in Table \ref{t6}.} \label{fig10}
\end{figure}

\begin{figure}[!ht]
\centering
\includegraphics[width=16cm, height=8cm]{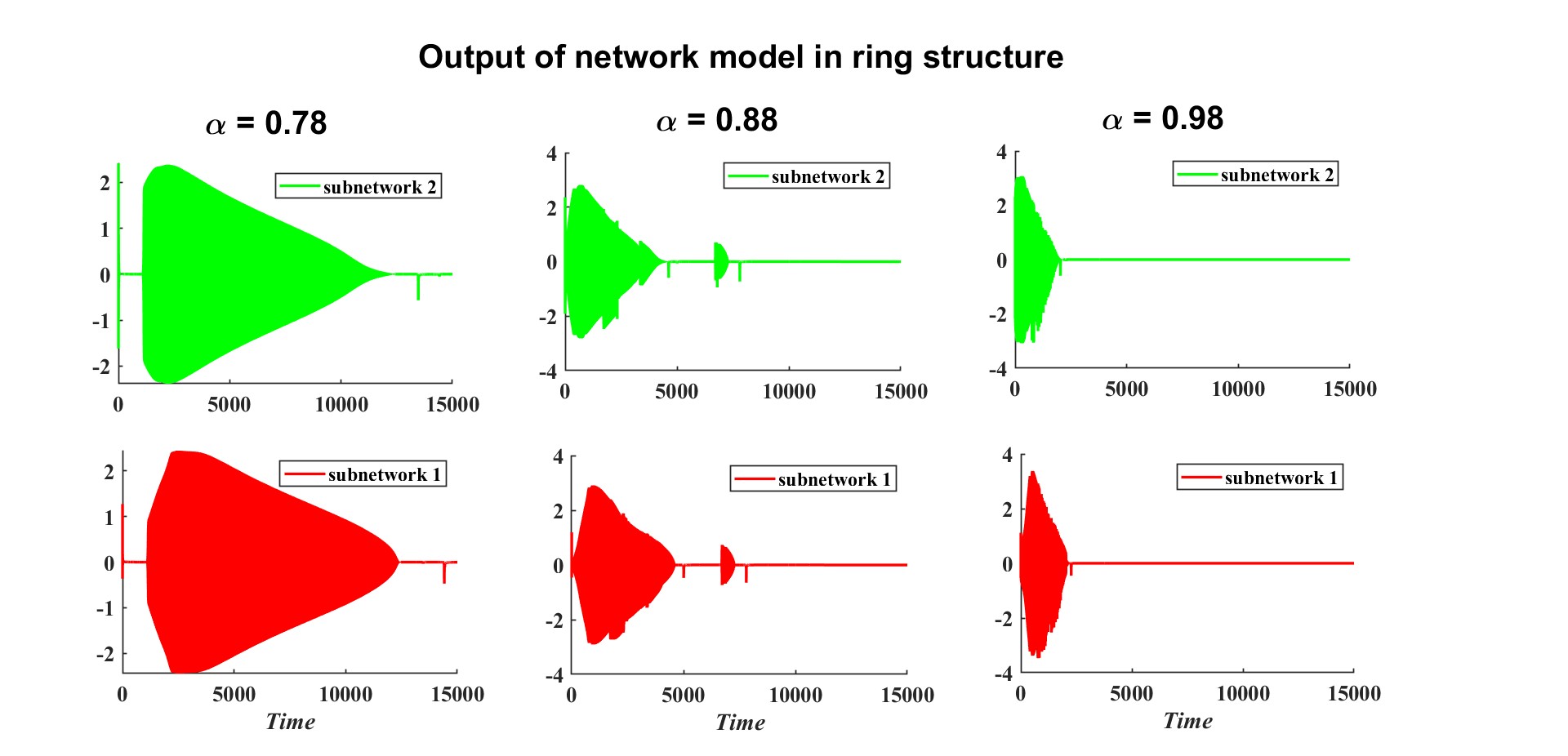}
\caption{The behavior over time of the model  (\ref{ndim}), respect to $X_1$,  with   $\alpha = 0.78$, $\alpha = 0.88$, $\alpha = 0.98.$ } \label{fig11}
\end{figure}

\section{Discussion and conclusion}
In neural network coupling channels, the electromagnetically induced current can be accurately represented using a flux-controlled memristor synapse, effectively modeling coupling dynamics. On the other hand, fractional-order models based on discrete maps are better suited for capturing transient effects such as ion pumping between cells, as fractional-order derivatives effectively describe memory effects in processing. Therefore, integrating memristor synapses with fractional-order discrete neuron models can more precisely estimate the biophysical dynamics of neural networks. This paper initially proposes a discrete fractional-order Rulkov neuron map utilizing Caputo’s nabla definition, involving two neurons coupled via two memristors. 
Subsequently, a novel stability theorem for nonlinear discrete fractional-order differential equations based on Caputo's nabla concept is established theoretically. Utilizing this theorem, equilibrium points of the proposed Rulkov neuron model are systematically analyzed for stability, determining the model's stability region. The boundaries of this stability region are defined according to the limits of the system's equilibrium points. We examined numerical methods in certain cases and validated the results.

Additionally, a network of neurons arranged in a ring structure is presented, and its stability region is similarly examined using the established theorem. The stability region boundaries are influenced by factors such as the fractional order, number of units, coupling strength, and number of neighbors in each subnetwork. It is demonstrated that adjusting the derivation order in both the initial model and the ring structure affects the time required to reach stability; specifically, lower derivation orders extend the time needed to achieve stability.


\end{document}